\documentclass[conference,10pt]{IEEEtran}
\IEEEoverridecommandlockouts

\usepackage{pdfsync}
\usepackage{xspace}
\usepackage{amsthm}
\usepackage{thmtools,thm-restate}
\usepackage{tikz}
\usepackage[mathscr]{euscript}
\usepackage{amsmath,amssymb}
\usepackage{hyperref}
\usetikzlibrary{decorations.pathreplacing,hobby,backgrounds,trees,arrows,automata,shapes,positioning,fit,patterns,calc,matrix}
\usepackage[scaled=1.]{dsserif}
\let\mathbbm\mathbb
\usepackage{latexcolors}
\usepackage{nowidow}
\usepackage{widows-and-orphans}
\usepackage[kerning=true,spacing=true]{microtype}

\newcommand{\overbar}[1]{\mkern 1.5mu\overline{\mkern-1.5mu#1\mkern-1.5mu}\mkern 1.5mu}

\newcommand{\nfa}{{{NFA}}\xspace}

\newcommand{\nfas}{{\textup{NFAs}}\xspace}

\newcommand{\nat}{\ensuremath{\mathbb{N}}\xspace}
\newcommand{\integ}{\ensuremath{\mathbb{Z}}\xspace}
\newcommand{\veps}{\ensuremath{\varepsilon}\xspace}

\newcommand{\As}{\ensuremath{\mathscr{A}}\xspace}
\newcommand{\Bs}{\ensuremath{\mathscr{B}}\xspace}
\newcommand{\Cs}{\ensuremath{\mathscr{C}}\xspace}

\newcommand{\Fs}{\ensuremath{\mathscr{F}}\xspace}
\newcommand{\Gs}{\ensuremath{\mathscr{G}}\xspace}

\newcommand{\Is}{\ensuremath{\mathscr{I}}\xspace}

\newcommand{\Ls}{\ensuremath{\mathscr{L}}\xspace}

\newcommand{\Hb}{\ensuremath{\mathbf{H}}\xspace}

\newcommand{\Kb}{\ensuremath{\mathbf{K}}\xspace}
\newcommand{\Lb}{\ensuremath{\mathbf{L}}\xspace}

\newcommand{\inv}{\ensuremath{^{-1}}}

\newcommand{\tila}{\ensuremath{\tilde{A}}\xspace}
\newcommand{\tilas}{\ensuremath{\tilde{A}^*}\xspace}

\newcommand{\folcl}[1]{\ensuremath{[#1]_{\simeq}}\xspace}
\newcommand{\fo}{\ensuremath{\textup{FO}}\xspace}

\newcommand{\setfqr}{\ensuremath{F_{\{q,r\}}}\xspace}

\newcommand{\frP}{\ensuremath{\mathbbm{P}}\xspace}
\newcommand{\prefsig}[1]{\ensuremath{\frP_{#1}}\xspace}

\newcommand{\sic}[1]{\ensuremath{\Sigma_{#1}}\xspace}
\newcommand{\bsc}[1]{\ensuremath{\Bs\Sigma_{#1}}\xspace}

\newcommand{\acons}[1]{\ensuremath{\langle #1 \rangle}\xspace}
\newcommand{\acauto}{\acons{\As}}
\newcommand{\acdel}{\acons{\delta}}

\newcommand{\econs}[1]{\ensuremath{\langle #1 \rangle_\veps}\xspace}
\newcommand{\ecauto}{\econs{\As}}
\newcommand{\ecdel}{\econs{\delta}}

\newcommand{\lauto}[3]{\ensuremath{L(#1,#2,#3)}\xspace}
\newcommand{\alauto}[2]{\lauto{\As}{#1}{#2}}

\newcommand{\aclauto}[2]{\lauto{\acauto}{#1}{#2}}
\newcommand{\aelauto}[2]{\lauto{\ecauto}{#1}{#2}}

\newcommand{\md}{\ensuremath{\textup{MOD}}\xspace}
\newcommand{\abg}{\ensuremath{\textup{AMT}}\xspace}
\newcommand{\grp}{\ensuremath{\textup{GR}}\xspace}

\newcommand{\ptime}{{\sc P}\xspace}
\newcommand{\nptime}{{\sc NP}\xspace}
\newcommand{\conptime}{co-{\sc NP}\xspace}

\newcommand{\nl}{\textup{NL}\xspace}
\newcommand{\conl}{\textup{co{-}NL}\xspace}
\newcommand{\pspace}{{\sc PSPACE}\xspace}

\newcommand{\tsat}{\textup{3-SAT}\xspace}

\tikzset{every state/.style={draw=indigo(dye)!80,very thick,fill=capri!20}}
\tikzset{statesub/.style={state,minimum size=.7cm,inner sep=1pt}}
\tikzset{pattstate/.style={state,draw=red!50!yellow,line width=2pt,fill=red!50!yellow!20}}
\tikzset{pdotstate/.style={state,minimum size=0.75cm,inner sep=0.5pt,draw=red!50!yellow,line
    width=2pt,dashed,fill=red!50!yellow!20}}
\tikzset{lstate/.style={state,minimum size=0.5cm,inner sep=0.5pt}}
\tikzstyle{trans}=[shorten >= 1pt,thick,->]

\makeatletter
\tikzstyle{initial by arrow}=   [after node path=
{
  {
    [to path=
    {
      [->,double=none,shorten >= 1pt,thick,every initial by arrow]
      ([shift=(\tikz@initial@angle:\tikz@initial@distance)]\tikztostart.\tikz@initial@angle)
      node [shape=coordinate,anchor=\tikz@initial@anchor,draw] {\tikz@initial@text}
      -- (\tikztostart)}]
    edge ()
  }
}]
\makeatother

\makeatletter
\tikzstyle{accepting by arrow}=   [after node path=
{
  {
    [to path=
    {
      [->,double=none,shorten >= 1pt,thick,every accepting by arrow]
      (\tikztostart) --
      ([shift=(\tikz@accepting@angle:\tikz@accepting@distance)]\tikztostart.\tikz@accepting@angle)
      node [shape=coordinate,anchor=\tikz@accepting@anchor,draw] {\tikz@accepting@text}
    }
    ]
    edge ()
  }
}]
\makeatother

\title{Group separation strikes back}

\author{\IEEEauthorblockN{Thomas Place}
  \IEEEauthorblockA{LaBRI, Univ. Bordeaux, CNRS, France\\
    Email: tplace@labri.fr}
  \and
  \IEEEauthorblockN{Marc Zeitoun}
  \IEEEauthorblockA{LaBRI, Univ. Bordeaux, CNRS, France\\
    Email: mz@labri.fr}}

\theoremstyle{plain}
\newtheorem{theorem}{Theorem}
\newtheorem{corollary}[theorem]{Corollary}
\newtheorem{fact}[theorem]{Fact}
\newtheorem{proposition}[theorem]{Proposition}
\newtheorem{lemma}[theorem]{Lemma}
\newtheorem{remark}[theorem]{Remark}
\newtheorem{example}[theorem]{Example}

\begin{document}

\maketitle

\begin{abstract}
  \emph{Group languages} are regular languages recognized by finite groups, or equivalently by finite automata in which each letter induces a permutation on the set of states. We investigate the separation problem for this class of languages: given two arbitrary regular languages as input, we show how to decide if there exists a group language containing the first one while being disjoint from the second. We prove that covering, a problem generalizing separation, is decidable. A simple covering algorithm was already known: it can be obtained indirectly as a corollary of an algebraic theorem by Ash. Unfortunately, while deducing the algorithm from this algebraic result is straightforward, all proofs of Ash's result itself require a strong background on algebraic concepts, and a wealth of technical machinery outside of automata theory. Our proof is independent of previous ones. It relies exclusively on standard notions from automata theory: we directly deal with separation and work with input languages represented by nondeterministic~finite~automata.

  We also investigate two strict  subclasses. First, the \emph{alphabet modulo testable languages} are those defined by counting the occurrences of each letter modulo some fixed integer (equivalently, they are the languages recognized by a commutative group). Secondly, the \emph{modulo languages} are those defined by counting the length of words modulo some fixed integer.  We prove that covering is decidable for both classes, with algorithms that rely on the construction made for group languages.

  Our proofs lead to tight complexity bounds for separation for all three classes, as well as for covering for both alphabet modulo testable languages and for modulo testable languages.
\end{abstract}

\section{Introduction}
\label{sec:intro}
\noindent
{\bf Context.} A prominent question in automata theory is to  understand natural classes of languages defined by restricting the common definitions of regular languages (such as regular expressions, automata, monadic second-order logic or finite monoids). Of course, ``understanding a class'' is an informal goal. The standard approach is to show that the class under investigation is recursive by looking for \emph{membership algorithms}: given a regular language as input, decide whether it belongs to the class. Rather than the procedure itself, the motivation is that formulating such an algorithm often requires a deep understanding of the class. This approach was initiated in the 60s by Schützenberger~\cite{schutzsf}, who provided a membership algorithm for the class of \emph{star-free languages} (those defined by a regular expression without Kleene star but with complement instead). This theorem started a~fruitful line of research, which is now supported by a wealth of results. In fact, some of the most famous open problems in automata theory are membership questions  (see~\cite{testher,jep-op35,jep-dd45} for surveys).

In this paper, we look at two problems, which both generalize membership. The first one is \emph{separation}: given \emph{two} regular languages $L_1$ and $L_2$ as input, decide whether there exists a third language that belongs to the investigated class, includes $L_1$ and is disjoint from~$L_2$. The second one is \emph{covering}. It generalizes separation to an arbitrary number of input languages. These problems have been getting a lot of attention recently, and one could even argue that they have replaced membership as the central question. The motivation is twofold. First, it has recently been shown~\cite{pzjacm19} that separation and covering are key ingredients for solving some of the most difficult membership questions (see~\cite{PZ:generic18} for a survey). Yet, the main motivation is tied to our original goal: ``understanding classes''. In this respect, separation and covering are more rewarding than membership (albeit more difficult). Intuitively, a membership algorithm for a class \Cs can only detect the languages in \Cs, while a covering algorithm provides information on how \emph{arbitrary regular languages} interact with~\Cs.

\medskip
\noindent
{\bf Group languages.} In the paper, we look at three specific classes. The main one is the class of \emph{group languages} \grp. While natural, this class is rather unique since its only known definition is based on machines: group languages are those recognized by a finite group, or equivalently, by a \emph{permutation automaton}~\cite{permauto} (a deterministic finite automaton in which each letter induces a permutation on the set of states). On the other hand, no ``descriptive'' definition of \grp is known (\emph{e.g.}, based on regular expressions or on logic). This makes it difficult to get an intuitive grasp about group languages, which may explain why this class remains poorly understood. We also consider two more intuitive subclasses: the first, \abg, consists of all languages recognized by \emph{Abelian} (\emph{i.e.}, commutative) groups. From a language theoretic point of view, these are the languages that can be defined by counting the occurrences of each letter modulo some fixed integer. The second is a subclass of \abg named \md. A language is in \md if membership of a word in the language only depends on its length modulo some fixed integer. Like all classes of group languages, these three classes are orthogonal and complementary to the classes for which separation and covering have been recently investigated (\emph{i.e.}, subclasses of the star-free languages, see~\cite{PZ:generic18}). Indeed, only the empty and universal languages are simultaneously star-free and group~languages.

\medskip
\noindent
\emph{\bf Motivations.} The class \grp and its subclasses serve as ingredients for building more~complex classes. This is well illustrated by logic: one may associate several classes to a fixed fragment of first-order logic. Each such class corresponds to a choice of \emph{signature} (\emph{i.e.}, the allowed predicates). For a class of languages~\Cs, define a signature \prefsig\Cs as follows: each language~$L$ in~\Cs gives rise to a predicate $P_L(x)$ selecting all positions~$x$ in a word $w$ such that the prefix of $w$ up to $x$ (excluded) belongs to~$L$. When $\Cs$ is $\abg$ or $\md$, we obtain two natural signatures: the predicates of \prefsig\abg allow one to test, for each letter $a$ of the alphabet, the number of $a$'s before position $x$ modulo some integer. Likewise, the predicates of \prefsig\md make it possible to test the value of positions modulo some integer.

More generally, given an arbitrary class \Gs of group languages, it is natural to consider the signatures $\{<\} \cup \prefsig\Gs$ and $\{<,{+1}\} \cup \prefsig\Gs$ (where ``${+}1$'' denotes the successor). It was recently shown that for many fragments of first-order logic \Fs, membership and sometimes even separation and covering are decidable for $\Fs(<,\prefsig\Gs)$ and $\Fs(<,{+1},\prefsig\Gs)$ as soon as \hbox{\emph{separation}} is decidable for \Gs. Prominent examples include the whole first-order logic~\cite{pzsfclos} (\fo), the first levels $\sic1$, $\bsc1$, $\sic2$ and $\sic3$ of the well-known quantifier alternation hierarchy of \fo~\cite{pzconcagroup,place2022characterizing,pzptimelevel1}, as well as two variable first-order logic ($\fo^2$)  and its whole quantifier alternation hierarchy~\cite{amazing2vars}. The proofs are based on language theoretic definitions of these classes, which are built by applying operators to \Gs. Consequently, it is desirable to have \emph{accessible language theoretic proofs}  that separation is decidable for the most prominent classes of group languages: \grp, \abg and~\md.

\medskip\noindent\textbf{Connection with other fields.} Separation by group languages is related to another area of independent interest: the decidability of separation for \grp and \abg can be deduced from purely algebraic results, which were proved even before separation was considered on the side of language theory. For example, the decidability of \grp separation follows from a theorem by Ash~\cite{Ash91}, who solved a conjecture of Rhodes~\cite{Karnofsky1982,typeii2} in semigroup theory. In this vast field of research, there are many publications on this topic, including several alternative proofs of Ash's theorem (see \emph{e.g.},~\cite{Herwig1999ExtendingPA},~\cite{Auinger04}, or~\cite{ribesz} relying on~\cite{pingroup}).

Ash's result spawned other lines of research in algebra. For instance, it motivated the computation of closures of regular languages in profinite topologies. Indeed, deciding whether such profinite closures intersect corresponds to deciding a property of subsets of finite monoids, which in turn is equivalent to deciding covering~\cite{almeidasep}. Extensions of such properties have been investigated for groups~(\emph{e.g.},~\cite{Benj_pgroups,almeida:groups}), Abelian groups~(\emph{e.g.},~\cite{DBLP:journals/ijac/AlmeidaD05,Alibabaei2019}) or other algebraic classes~(\emph{e.g.},~\cite{almeida05:R}).

However, this line of research is disconnected from our motivation: to obtain a direct and purely automata theoretic proof of the decidability of \grp-covering, in order to understand the involved combinatorics on regular languages. In particular, we should not rely on Ash's result itself (unlike some of the work cited above).   Unfortunately, the existing proofs of Ash's result do not meet our motivation. Indeed, they do not involve covering. Their use therefore requires a detour: abstract the problem as a purely algebraic/topological question, do the proof in this framework and then come back to covering.

From our perspective, this detour has several drawbacks. First, it relies on ``black box'' results: to get a complete proof requires to gather and understand a lot of material. Secondly, the proofs demand a solid background on algebraic concepts and a wealth of technical machinery outside  automata theory: for instance, for Ash's theorem, some proofs~\cite{Ash91,Auinger04} are based on the theory of inverse semigroups while others rely on topological arguments~\cite{pingroup,MARGOLIS1992305,ribesz}.  For these reasons, beyond their intrinsic difficulty, these proofs do not bring any intuition \emph{at the level of languages}.  This means that these results and their proofs are not satisfactory with respect to our primary objective: ``understanding~classes of languages''.

It has been shown that this detour can be avoided for almost all natural classes~\cite{amazing2vars,pzptimelevel1,place2022characterizing,pzconcagroup}: one can work directly with languages and use only basic algebraic notions (typically, the definition of regular languages by morphisms into finite monoids and standard combinatorics on monoids). This direct approach is much more rewarding with respect to our original goal. In particular, the proofs of separation algorithms provide an explicit description of generic separators (when they exist).

Group languages are among the few classes for which it is not known whether a fully language-theoretic approach is possible. This is the question  we address in this paper.

\medskip
\noindent
\emph{\bf Contributions.} We present self-contained proofs that covering and separation are decidable for \grp, \abg and \md. They avoid the detour through algebra and are based on new ideas that are independent of any pre-existing indirect proofs in this area. Of course, our proofs remain involved: these are hard questions. However, they rely exclusively on basic notions of automata theory, which makes them accessible to computer scientists. We work with nondeterministic finite automata (\nfa).  Paradoxically, we use very few algebraic notions beyond the standard definition of a group. Roughly speaking, proofs are based on word combinatorics for \grp, on arithmetic for \abg, while \md reduces to the other two for unary alphabets.

All separation and covering algorithms are neat and simple. They work directly with input languages represented by \nfas. However, it is worth mentioning that the main value of the paper lies not in the algorithms themselves, but in their proofs. Indeed, is actually easy to derive these algorithms from the aforementioned independent algebraic results. In particular, the covering algorithm we present for \grp is essentially a reformulation on automata and a simple corollary of the original algorithm obtained from Ash's theorem~\cite{Ash91}, which uses inputs represented by monoid morphisms rather than automata. Actually, we show how to deduce our algorithm from Ash's one. Furthermore, an algorithm similar to ours is given in~\cite{pingroup}. It relies on a conjecture proved later in~\cite{ribesz}, and on an algorithm to compute closures of certain regular languages in an appropriate topology~\cite{BirgetMargolisMeakinWeil,Myasnikov-Kapovich}, itself based on Stallings foldings~\cite{stallings}. In contrast, our new proof is direct, matching our original objective: to remain in the framework of automata throughout the whole argument. In fact, \nfas are a key ingredient of this proof: we use nondeterminism in a crucial~way.

Let us illustrate the simplicity of the algorithms using \grp. We present a simple construction that inputs an \nfa \As and outputs a new one \ecauto. Then, we show that the languages recognized by two \nfas $\As_1$ and $\As_2$ can be separated by a group language if and only if the languages recognized by $\econs{\As_1}$ and~$\econs{\As_2}$ do \emph{not} intersect. Since $\econs{\As_1}$ and~$\econs{\As_2}$ can be computed in polynomial time, this shows that \grp-separation is in~\ptime (this goes up to \pspace for covering as this boils down to deciding intersection between an arbitrary number of \nfas). The approach for \abg is similar with one key difference: we look at \emph{Parikh images}. More precisely, we show that whether the languages recognized by two \nfas $\As_1$ and $\As_2$ can be separated by \abg boils down to some specific condition on the Parikh images of $\econs{\As_1}$ and $\econs{\As_2}$. The standard result that states that existential Presburger arithmetic is in \nptime~\cite{presNP1} implies then that \abg-separation is in \conptime. Actually, we show that both \abg-separation and \abg-covering are \conptime-complete. Finally, we show that in the much simpler case of \md, separation is \nl-complete and covering  \conptime-complete.

\medskip\noindent
\emph{\bf Organization.} In Section~\ref{sec:prelims}, we introduce preliminary definitions and a key automata construction used in all algorithms. Section~\ref{sec:grp} is devoted to separation and covering for the class \grp of all group languages (in particular, in Section~\ref{sec:history}, we show how to deduce our algorithm from Ash's original one). Section~\ref{sec:abe} is devoted to covering for \abg. Finally, Section~\ref{sec:mod} is devoted to covering for \md.

\medskip\noindent This paper is the full version of~\cite{pzgroups23}.

\section{Preliminaries}
\label{sec:prelims}

\subsection{Words, languages, separation and covering}

\noindent
{\bf Languages.} We fix an arbitrary finite alphabet $A$ for the paper. As usual, $A^*$ denotes the set of all finite words over $A$, including the empty word~\veps. We let $A^{+}=A^{*}\setminus\{\veps\}$. For $u,v \in A^*$, we let $uv$ be the word obtained by concatenating $u$ and~$v$. A \emph{language}  (over~$A$) is a subset of~$A^*$. Finally, a \emph{class of languages} \Cs is a set of languages, \emph{i.e.}, a subset of $2^{A^*}$. Additionally, we say that \Cs is a \emph{Boolean algebra} when it is closed under union, intersection and complement: for every $K,L \in \Cs$, we have $K \cup L \in \Cs$, $K \cap L \in \Cs$ and $A^* \setminus K \in \Cs$. In this paper, we consider \emph{regular languages}: those that can be equivalently defined by finite automata, finite monoids or monadic second-order logic. We work with the definition based on automata. We shall also consider monoids in order to connect our statements to historical results. Let us recall these two definitions.

\medskip\noindent {\bf Automata.} A nondeterministic finite automaton (\nfa) over~$A$ is a tuple $\As= (Q,I,F,\delta)$ where $Q$ is a finite set of states, $I \subseteq Q$ and $F \subseteq Q$ are sets of initial and final states, and $\delta \subseteq Q \times A \times Q$ is a set of transitions. We define the language recognized by \As, denoted by $L(\As)$, as follows. Given $q,r \in Q$ and $w\in A^*$, we say that there exists a \emph{run labeled by $w$ from $q$ to $r$} (in \As) if there exist $q_0,\dots,q_n \in Q$ and $a_1,\dots,a_n \in A$ such that $w = a_1 \cdots a_n$ , $q_0 = q$, $q_n = r$ and $(q_{i-1},a_i,q_i) \in \delta$ for every $1 \leq i \leq n$. Given $q,r \in Q$, we write $\alauto{q}{r}$ for the language consisting of all words $w \in A^*$ such that there exists a run labeled by $w$ from $q$ to $r$ (note that $\veps \in \alauto{q}{q}$ for every $q \in Q$). The \emph{language $L(\As)$ recognized by \As} is $\bigcup_{q \in I} \bigcup_{r \in F} \alauto{q}{r}$. We say that a language is \emph{regular} when it is recognized by an~\nfa.

We also consider \nfas with \veps-transitions. In such an \nfa $\As = (Q,I,F,\delta)$, a transition may also be labeled by the empty word ``\veps'' (that is, $\delta\subseteq Q \times (A \cup \{\veps\} )\times Q$). We use the standard semantics: an \veps-transition can be taken without consuming an input letter. Unless otherwise specified, the \nfas that we consider are assumed to be  \emph{without}~\veps-transitions.

In the examples, we represent \nfas by graphs, as usual: nodes are the states (the initial states have an incoming arrow and the final ones an outgoing arrow), and each transition $(q,a,r)$ is depicted by an $a$-labeled edge from $q$ to $r$.

\medskip\noindent
{\bf Monoids.} A \emph{monoid} is a set $M$ endowed with an associative multiplication $(s,t)\mapsto st$ having an identity element $1_M$, \emph{i.e.}, such that ${1_M}s=s {1_M}=s$ for every $s \in M$. Clearly, $A^{*}$ is a monoid whose multiplication is concatenation (the identity element is~\veps). Therefore, we may consider monoid morphisms $\alpha: A^* \to M$ where $M$ is an arbitrary monoid: they are the mappings satisfying $\alpha(\veps)=1_M$ and $\alpha(uv)=\alpha(u)\alpha(v)$ for all $u,v\in A^*$. Given such a morphism and some language $L \subseteq A^*$, we say that $L$ is \emph{recognized} by~$\alpha$ when there exists a set $F \subseteq M$ such that $L = \alpha\inv(F)$. It is well-known and simple to verify that a language is regular if and only if it is recognized by a morphism into a \emph{finite} monoid.

\subsection{Separation and covering}

We now define two decision problems, which depend on an arbitrary fixed class~\Cs. They are used as mathematical tools for investigating \Cs. They take finitely many regular languages as input (which we represent with \nfas in the paper).

Given two languages $L_1,L_2$, we say that $L_1$ is \emph{\Cs-separable} from $L_2$ if there exists $K\! \in\! \Cs$ such that $L_1\! \subseteq\! K$ and $L_2 \cap K\! =\! \emptyset$. The \emph{\Cs-separation problem} takes two regular languages $L_1$ and $L_2$ as input and asks whether $L_1$  is \Cs-separable from $L_2$.

Covering is a generalization introduced in~\cite{pzcovering2}. Given a language~$L$, a \emph{\Cs-cover of $L$} is a \emph{finite} set of languages \Kb such that every $K \in \Kb$ belongs to \Cs and~$L \subseteq \bigcup_{K \in \Kb} K$. Given a pair $(L_1,\Lb_2)$ where $L_1$ is a language and $\Lb_2$ a \emph{finite set of languages}, we say that $(L_1,\Lb_2)$ is \emph{\Cs-coverable} when there exists a \Cs-cover \Kb of $L_1$ such that for every $K\in\Kb$, there exists $L \in \Lb_2$ satisfying $K \cap L = \emptyset$. The \emph{\Cs-covering problem} takes as input a regular language~$L_1$ and a finite set of regular languages $\Lb_2$ and asks whether $(L_1,\Lb_2)$ is~\Cs-coverable.

Covering generalizes separation when \Cs is closed under union: in this case, one may verify that $L_1$ is \Cs-separable from $L_2$, if and only if $(L_1,\{L_2\})$ is \Cs-coverable. Additionally, the definition of covering may be simplified when \Cs is a \emph{Boolean algebra}: it suffices to consider the case when the language $L_1$ that needs to be covered is $A^*$. Indeed, in that case, $(L_1,\Lb_2)$ is \Cs-coverable if and only if $(A^*,\{L_1\} \cup \Lb_2)$ is \Cs-coverable (the proof is simple, see~\cite{pzcovering2}).

We say that  a finite set of languages \Lb is \emph{\Cs-coverable} when $(A^*,\Lb)$ is \Cs-coverable. If \Cs is a Boolean algebra, the above remark shows that \Cs-covering boils down to deciding whether a finite input set \Lb of regular languages is \Cs-coverable~\cite{pzcovering2}. Also, \Cs-separation is the special case when $|\Lb| = 2$.

\begin{remark}
  When discussing complexity, we consider the alphabet~$A$ as part of the~input.
\end{remark}

\subsection{Group languages}

A \emph{group} is a monoid $G$ such that every element $g \in G$ has an inverse $g\inv \in G$, \emph{i.e.}, $gg\inv = g\inv g = 1_G$. We write \grp for the class of all \emph{group languages}, \emph{i.e.}, which are recognized by a  morphism into a \emph{finite group}. One can verify that \grp is a Boolean algebra.

\begin{remark}
  No language theoretic definition of \grp is known (\emph{i.e.}, by specific regular expressions). There is however an automata-based definition: group languages are those recognized by permutation automata (\emph{i.e.}, which are simultaneously deterministic, co-deterministic and complete). For instance, $(ab^*a+ba^*b)^*$ is a group language. Indeed, it is recognized by the permutation automaton drawn below, and by the morphism into the symmetric group on $\{1,2,3\}$ that maps $a$ to the transposition $(1,2)$ and $b$ to the transposition $(1,3)$.
  
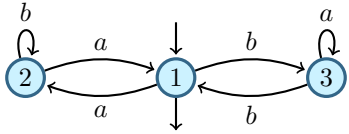
\begin{center}
    \begin{tikzpicture}
      \node[lstate,initial above,accepting below] (q0) at (0,0) {$1$};
      \node[lstate] (q1) at (-2.0,0) {$2$};
      \node[lstate] (q2) at (2.0,0) {$3$};

      \draw[trans] (q1) edge [loop above] node[above] {$b$} (q1);
      \draw[trans] (q1) edge [bend left=20] node[above] {$a$} (q0);
      \draw[trans] (q0) edge [bend left=20] node[below] {$a$} (q1);

      \draw[trans] (q0) edge [bend left=20] node[above] {$b$} (q2);
      \draw[trans] (q2) edge [bend left=20] node[below] {$b$} (q0);
      \draw[trans] (q2) edge [loop above] node[above] {$a$} (q2);
    \end{tikzpicture}
  \end{center}
\end{remark}

We also look at two subclasses. The first one is the class \md of \emph{modulo languages}. For $w \in A^*$, we write $|w| \in \nat$ for the \emph{length} of $w$ (its number of letters). For all $q,r \in \nat$ such that $r < q$, we let $L_{q,r} = \{w \in A^* \mid |w| \equiv r \bmod q\}$. The class \md consists of all \emph{finite unions} of languages $L_{q,r}$. We turn to the class \abg of \emph{alphabet modulo testable languages}. If $w \in A^*$ and $a \in A$, let~$|w|_a \in \nat$ be the number of copies of ``$a$'' in $w$. For all $q,r \in \nat$ such that $r < q$ and all $a \in A$, let $L^a_{q,r} = \{w \in A^* \mid |w|_a \equiv r \bmod q\}$. We let \abg be the least class containing all languages $L^a_{q,r}$ and closed under union and intersection. It can be verified that both \md and \abg are Boolean algebras and that $\md \subsetneq \abg \subsetneq \grp$. In the paper, we prove that covering and separation are decidable for \grp, \abg and \md. The proofs are based exclusively on elementary arguments from automata theory. We rely on a common automata-based construction, which we now present.

\subsection{Automata-based construction}

We extend $A$ as a larger alphabet denoted by \tila. For each $a \in A$, we create a fresh letter $a\inv$ (by ``fresh'', we mean that $a\inv \not\in A$) and define $A\inv = \{a\inv \mid a \in A\}$. We let \tila be the disjoint union $\tila= A \cup A\inv$. Observe that we have a bijection $a \mapsto a\inv$ from $A$ to $A\inv$. We extend it as an involution of \tilas: for every $a \in A$, we let $(a\inv) \inv = a$. Then, for every $w = b_1b_2 \cdots b_n \in \tilas$ (with $b_1 ,\dots,b_n \in \tila)$, we define $w\inv = b_n\inv \cdots b_2\inv b_1\inv$ (we let $\veps\inv = \veps$). The map $w \mapsto w\inv$ is an involution of \tilas: $(w\inv)\inv = w$.

Every morphism $\alpha: A^* \to G$ into a group $G$ can be extended as morphism $\alpha:\tilas\to G$. For all $a\inv \in A\inv$, we let $\alpha(a\inv) = (\alpha(a))\inv$ (\emph{i.e.}, $\alpha(a\inv)$ is the inverse of $\alpha(a)$ in~$G$).   One may verify that the definition implies $\alpha(w\inv) = (\alpha(w))\inv$ for every $w\in\tilas$. We shall use this fact implicitly.

\begin{remark}
  This construction is standard, and used to introduce the free group over $A$ (which is a quotient of \tilas).  We do not need this notion. We use \tila as a syntactic tool: we build auxiliary \nfas over \tila from \nfas over $A$. We shall never consider arbitrary objects over~\tila: all arbitrary \nfas that we encounter are implicitly assumed to be over $A$.
\end{remark}

We turn to the main construction. Let $\As = (Q,I,F,\delta)$ be an arbitrary \nfa over the original alphabet $A$ (\emph{i.e.},
$\delta \subseteq Q \times A \times Q$). We build a new \nfa \acauto over the extended alphabet \tila. We say that two states $q,r \in Q$ are \emph{strongly connected} if we have $\alauto{q}{r} \neq \emptyset$ and $\alauto{r}{q} \neq\emptyset$  (\emph{i.e.}, $q$ and~$r$ are in the same strongly connected component of the graph representation of~\As). This is an equivalence relation. We let $\acdel \subseteq Q\times \tila \times Q$ as the following extended set of~transitions:
\[
  \acdel = \delta \cup \left\{(r,a\inv,q) \mid
    \begin{array}{l}
      (q,a,r) \in \delta \text{ and} \\
      \text{\,$q,r$ are strongly connected}
    \end{array}\right\}. \]
We let $\acauto = (Q,I,F,\acdel)$, so that $L(\acauto) \subseteq \tilas$. Observe that for all $u \in \tilas$ and all strongly connected $q,r \in Q$, we have $u \in \aclauto{q}{r}$ if and only if $u\inv \in \aclauto{r}{q}$. Moreover, we can compute \acauto from \As in polynomial time: this boils down to computing the pairs of strongly connected states, \emph{i.e.}, to directed graph reachability. The following lemma is used to ``simulate'' the runs in \acauto into the original~\nfa~\As.

\begin{lemma} \label{lem:trans}
  Let $\As = (Q,I,F,\delta)$ be an \nfa and $\alpha: A^*\to G$ be a morphism into a finite group. For every $q,r \in Q$ and $w \in \aclauto{q}{r} \subseteq \tilas$, there exists a word $w' \in A^*$ such that $w' \in \alauto{q}{r}$ and $\alpha(w) = \alpha(w')$.
\end{lemma}

\begin{proof}
  We have to show that for every \mbox{$(s,a\inv,t) \in \acdel$} where $a \in A$, there exists $x \in \alauto{s}{t}$ such that $\alpha(x) = (\alpha(a))\inv$. By definition of \acdel, we have $(t,a,s)\in\delta$ and $s,t$ are strongly connected. Hence, we get $y \in  \alauto{s}{t}$. Since $G$ is a \emph{finite} group, it is standard that there exists $p \geq 1$ such that $g^p=1_G$ for all $g\in G$. Thus, $\alpha((ay)^p) = 1_G$. Let $x = y(ay)^{p-1}$. By hypothesis on $a$ and $y$, we know that $x \in \alauto{s}{t}$. Since $\alpha(ax) = \alpha((ay)^p) = 1_G$, we obtain $\alpha(x) = (\alpha(a))\inv$, as~desired.
\end{proof}

\section{Covering for group languages}
\label{sec:grp}
We prove that separation and covering are decidable for~\grp. Historically, this was first obtained as a corollary of a difficult independent algebraic theorem by Ash~\cite{Ash91}  (see~\cite{henckell:hal-00019815} for details and~\cite{almeidasep} for the link with separation). Our algorithm is essentially the one obtained from this theorem. Yet, we choose a different presentation: our inputs are represented by \nfas whereas the original algorithm considers a single monoid morphism recognizing all inputs. In itself, the \nfa-based procedure is merely a natural reformulation of the one based on monoids (see Section~\ref{sec:history} for details). On the other hand, we work exclusively with \nfas, which is a new idea, and nondeterminism is a key ingredient in our proof.

This proof is our main contribution. The known proofs of Ash's theorem (\emph{e.g.}, \cite{Ash91,Auinger04,ribesz,Herwig1999ExtendingPA}) are arduous. Typically, they rely on specialized notions from independent fields such as algebra, topology or model theory. Moreover, they use black box results. In contrast, our proof is direct. While still difficult, it is fully self-contained and relies only on elementary notions from automata theory and combinatorics on words.

\subsection{Statement}

The procedure is based on a theorem characterizing the finite sets of regular languages that are \grp-coverable. We first extend the core construction $\As \mapsto \acauto$ introduced in the previous section (this extension is specific to~\grp-covering).

Given an arbitrary \nfa \As, we further modify the \nfa \acauto and construct a new \nfa \emph{with \veps-transitions} \ecauto (these are the only \nfas with \veps-transitions that we consider). The definition is based on a language $L_\veps \subseteq \tilas$ that we define first. We introduce a standard rewriting rule that one may apply to words in \tilas. If $w \in\tilas$ contains an infix of the form $aa\inv$ or $a\inv a$ for some $a\in A$, one may delete it. More precisely, given $w, w' \in \tilas$, we write $w \rightarrow w'$ if there exist $x,y \in \tilas$ and $a \in A$ such that either $w = xaa\inv y$ or $w =xa\inv a y$, and $w' = xy$. We write ``$\xrightarrow{*}$'' for the reflexive transitive closure of ``$\rightarrow$''. That is, given $w,w' \in \tilas$, we have $w\xrightarrow{*} w'$ if $w = w'$ or there exist words $w_0,\dots,w_n \in \tilas$ with $n \geq 1$ such that $w = w_0 \rightarrow w_1 \rightarrow w_2  \rightarrow \cdots \rightarrow w_n = w'$. We let $L_\veps = \{w \in \tilas \mid w \xrightarrow{*} \veps\}$. This is a variant of the well-known \emph{Dyck language} which is \emph{not} regular (it is only context-free).

Consider an \nfa $\As = (Q,I,F,\delta)$ and the associated \nfa $\acauto = (Q,I,F,\acdel)$. We extend \acdel with \veps-transitions. We define $\ecdel \subseteq Q \times (\tila \cup \{\veps\}) \cup Q$ as follows:
\[
  \ecdel = \acdel \cup \bigl\{(q,\veps,r) \mid q,r \in Q \text{ and }L_\veps \cap \aclauto{q}{r} \neq \emptyset\bigr\}. \]
Moreover, we let $\ecauto = (Q,I,F,\ecdel)$.

\begin{example}\label{exa:sep}
  Let $L_1 = b(ab)^*$ and $L_2 = aa^*$. These languages are recognized by the following two \nfas $\As_1$ and $\As_2$:

  
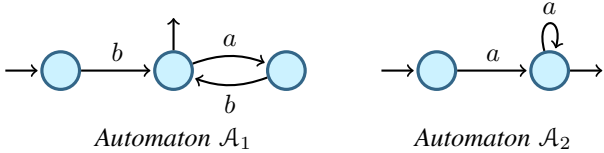
\begin{tikzpicture}
    \node[lstate,initial left] (q0) at (0,0) {};
    \node[lstate,accepting above] (q1) at (1.5,0) {};
    \node[lstate] (q2) at (3,0) {};
    \node at (1.5,-.9) {Automaton $\As_1$};

    \draw[trans] (q0) edge node[above] {$b$} (q1);
    \draw[trans] (q1) edge [bend left=20] node[above] {$a$} (q2);
    \draw[trans] (q2) edge [bend left=20] node[below] {$b$} (q1);

    \node[lstate,initial left] (r0) at (5,0) {};
    \node[lstate,accepting right] (r1) at (6.5,0) {};
    \node at (5.75,-.9) {Automaton $\As_2$};

    \draw[trans] (r0) edge node[above] {$a$} (r1);
    \draw[trans] (r1) edge [loop above] node[above] {$a$} (r1);
  \end{tikzpicture}

  \noindent
  We compute \econs{\As_1} and \econs{\As_2} (omitting \veps-labeled self-loops).

  
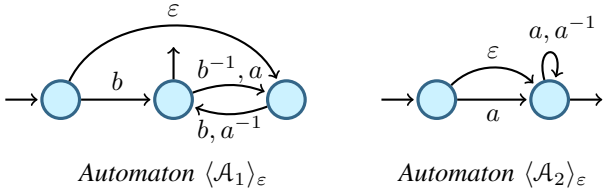
\begin{tikzpicture}
    \node[lstate,initial left] (q0) at (0,0) {};
    \node[lstate,accepting above] (q1) at (1.5,0) {};
    \node[lstate] (q2) at (3,0) {};
    \node at (1.5,-1) {Automaton $\econs{\As_1}$};

    \draw[trans] (q0) edge node[above] {$b$} (q1);
    \draw[trans] (q1) edge [bend left=20] node[above=-0.1] {$b\inv,a$} (q2);
    \draw[trans] (q2) edge [bend left=20] node[below=-0.1] {$b,a\inv$} (q1);
    \draw[trans] (q0) edge [bend left=65] node[above] {$\veps$} (q2);

    \node[lstate,initial left] (r0) at (5,0) {};
    \node[lstate,accepting right] (r1) at (6.5,0) {};
    \node at (5.75,-1) {Automaton $\econs{\As_2}$};

    \draw[trans] (r0) edge node[below] {$a$} (r1);
    \draw[trans] (r1) edge [loop above] node[above] {\quad$a,a\inv$} (r1);
    \draw[trans] (r0) edge [bend left=45] node[above] {$\veps$} (r1);
  \end{tikzpicture}
\end{example}

Note that one may compute \ecauto from \acauto (hence from \As) in \emph{polynomial time}. Indeed, the construction creates a new \veps-transition $(q,\veps,r)$ if and only if $\aclauto{q}{r}$ (which is regular) intersects $L_\veps$ (which is context-free). It is standard that this problem is decidable in polynomial time~\cite{BarHillel61}. We complete the definition with two simple but useful properties.

\begin{restatable}{fact}{fjump}\label{fct:jump}
  Let $\As= (Q,I,F,\delta)$ be an \nfa, $q,r \in Q$ and $w \in \tilas$ such that $w \in \aelauto{q}{r}$. If $w  \in L_\veps$, then $(q,\veps,r) \in \ecdel$. Also, if $q,r$ are strongly connected, then $w\inv \in \aelauto{r}{q}$.
\end{restatable}

\begin{proof}
  Since $w \in \aelauto{q}{r}$, the definition of $\ecauto$ yields $u_0,\dots,u_n\in \tilas$, $v_1,\dots,v_n\in L_\veps$ and $x \in \aclauto{q}{r}$ such that $w = u_0 \cdots u_n$ and $x = u_0v_1u_1 \cdots v_nu_n$. Assume first that $w \in L_\veps$. Since $v_1,\dots,v_n\in L_\veps$, we have $v_i \xrightarrow{*} \veps$ for every $i \leq n$. Hence, $x \xrightarrow{*} w$ and since $w \in L_\veps$, we get $x \xrightarrow{*} \veps$. Since $x \in \aclauto{q}{r}$, we get $(q,\veps,r) \in \ecdel$ by definition. Assume now that $q,r$ are strongly connected. Since $x \in \aclauto{q}{r}$, we get $x\inv \in \aclauto{r}{q}$ by definition of \acauto. Moreover, $x\inv = u_n\inv v_n\inv  \cdots u_1\inv v_1\inv u_0\inv$ and since $v_i \xrightarrow{*} \veps$ for all $i \leq n$, we have $v_i\inv \xrightarrow{*} \veps$ for all $i \leq n$. Thus, $x\inv \in \aclauto{r}{q}$ implies that $u_n\inv  \cdots u_0\inv \in \aelauto{r}{q}$. Since $u_n\inv  \cdots u_0\inv = w\inv$, this completes the proof.
\end{proof}

\noindent
Let us now extend Lemma~\ref{lem:trans} to this new automaton \ecauto.

\begin{restatable}{lemma}{jtrans}\label{lem:jtrans}
  Let $\As = (Q,I,F,\delta)$ be an \nfa and \mbox{$\alpha: A^*\to G$} be a morphism into a finite group. For every $q,r \in Q$ and \mbox{$w  \in \aelauto{q}{r} \subseteq \tilas$}, there exists a word $w' \in A^*$ such that $w' \in \alauto{q}{r}$ and $\alpha(w) = \alpha(w')$.
\end{restatable}

\begin{proof}
  As $w  \in \aelauto{q}{r}$, there are \mbox{$u_0,\dots,u_n \in \tilas$} and $v_1,\dots,v_n\in L_\veps$ such that $x=u_0v_1u_1 \cdots v_nu_n \in \aclauto{q}{r}$ and $w = u_0 \cdots u_n$. Lemma~\ref{lem:trans} yields a word $w' \in A^*$ such that $w' \in \alauto{q}{r}$ and $\alpha(x) = \alpha(w')$. Moreover, since $v_1,\dots,v_n\in L_\veps$, we have $v_i \xrightarrow{*} \veps$ for all $i \leq n$. Hence, $\alpha(v_i) = 1_G$. We get $\alpha(w) = \alpha(x) = \alpha(w')$, as~desired.
\end{proof}

We now state the main theorem. It characterizes \grp-coverability using the construction $\As \mapsto \ecauto$.

\begin{restatable}{theorem}{grpsep} \label{thm:grpsep}
  Let $k \geq 1$ and let $\As_1,\dots,\As_k$ be \nfas. The following conditions are equivalent:
  \begin{enumerate}
    \item The set $\big\{L(\As_1),\dots,L(\As_k)\big\}$ is \grp-coverable.
    \item We have $\bigcap_{i \leq k} L(\econs{\As_i}) = \emptyset$.
  \end{enumerate}
\end{restatable}

Clearly, the second condition in Theorem~\ref{thm:grpsep} can be decided. Indeed, for every $i \leq k$, we can compute $\econs{\As_i}$ from $\As_i$ in polynomial time. Moreover, it one can decide whether an arbitrary number of \nfas intersect (in polynomial space). Hence, we obtain as desired that \grp-covering is decidable and in \pspace (it is unknown whether this is tight). Additionally, when the number $k$ of inputs is fixed, intersection can be decided in polynomial \emph{time}. In particular, \grp-separation (the case $k =2$) is in~\ptime. We prove at the end of the section that the problem is \ptime-complete.

\begin{example}
  Recall the languages from Example~\ref{exa:sep}. Observe that $a\inv\in L(\econs{\As_1})\cap L(\econs{\As_2})$.  We deduce from Theorem~\ref{thm:grpsep} that no group language can separate $L_1$ from~$L_2$.
\end{example}

\subsection{Proof argument}

We fix a number $k \geq 1$ and for every $j \leq k$, we consider an \nfa $\As_j = (Q_j,I_j,F_j,\delta_j)$. The two implications in the theorem are handled independently. Let us start with $1) \Rightarrow 2)$.

\smallskip
\noindent
{\bf Implication $1) \Rightarrow 2)$.} We prove the contrapositive. Assume that there exists $w \in \bigcap_{j \leq k} L(\econs{\As_j})$. We prove that $\{L(\As_1),\dots,L(\As_k)\}$ is \emph{not} \grp-coverable. Hence, we fix an arbitrary \grp-cover \Kb of $A^*$ and exhibit $K \in \Kb$ such that $K \cap L(\As_j) \neq \emptyset$ for every $j \leq k$.

For all $i \leq n$, let $\alpha_i: A^* \to G_i$ be a morphism into a finite group recognizing $K_i \in \grp$. Clearly, $G = G_1 \times \cdots \times G_n$ is a finite group for the componentwise multiplication and the morphism $\alpha: A^* \to G$ defined by $\alpha(w) = (\alpha_1(w),\dots,\alpha_n(w))$ recognizes all languages $K_i$. Since $w \in L(\econs{\As_j})$ for every $j\leq k$, Lemma~\ref{lem:jtrans} yields $w_j \in A^*$ such that $w_j\in L(\As_j)$ and $\alpha(w_j) = \alpha(w)$. Since \Kb is a cover of $A^*$, there exists $K \in \Kb$ such that $w_1 \in K$. Hence, since $K$ is recognized by $\alpha$ and $\alpha(w_1)=\cdots = \alpha(w_k) = \alpha(w)$, we get $w_1,\dots,w_k \in K$. Thus, $K \cap L(\As_j) \neq \emptyset$ for all $j \leq k$, as desired.

\smallskip
\noindent
{\bf Implication $2) \Rightarrow 1)$.}  Let $Q = \bigcup_{j \leq k} Q_j$ (we assume without loss of generality that the sets $Q_j$ are pairwise disjoint) and $\delta = \bigcup_{j \leq k} \delta_j$. Let $\As = (Q,\emptyset,\emptyset,\delta)$. A \emph{synchronizer} (for \As) is a morphism $\alpha: A^* \to G$ into a finite group $G$ such that for all $n \geq 1$ and $q_1,\dots,q_n, r_1,\dots,r_n \in Q$, if there exists $g \in G$ such that $\alpha\inv(g) \cap \alauto{q_j}{r_j} \neq \emptyset$ for all $j\leq n$, then there exists $u \in \tilas$ such that $u \in \aelauto{q_j}{r_j}$ for all~$j \leq n$.

\begin{proposition}\label{prop:covesynch}
  There exists a synchronizer for \As.
\end{proposition}

We first use this result to prove $2) \Rightarrow 1)$. Assume that $\bigcap_{j \leq k} L(\econs{\As_j})\!=\! \emptyset$. By Proposition~\ref{prop:covesynch}, there exists  a synchronizer $\alpha\!:\! A^*\! \to\! G$ for \As. Let $\Kb = \{\alpha\inv(g) \mid g \in G\}$, which is a \grp-cover of $A^*$. We show that for every $g \in G$, there exists $j \leq k$ such that $\alpha\inv(g) \cap L(\As_j) = \emptyset$. This implies as desired that $\{L(\As_1),\cdots,L(\As_k)\}$ is \grp-coverable. Let $g \in G$. By contradiction, assume that $\alpha\inv(g) \cap L(\As_j) \neq \emptyset$ for every $j \leq k$. For each $j \leq k$, this yields $q_j \in I_j$ and $r_j \in F_j$ such that $\alpha\inv(g) \cap \alauto{q_j}{r_j} \neq \emptyset$. Since $\alpha$ is a synchronizer for \As, we obtain $u \in \tilas$ such that $u \in \aelauto{q_j}{r_j}$ for every $j \leq k$. Since $q_j \in I_j$ and $r_j \in F_j$, it follows that $u \in L(\econs{\As_j})$ for every $j \leq k$, contradicting the hypothesis that $\bigcap_{j \leq k} L(\econs{\As_j}) = \emptyset$. This concludes the main argument.

\smallskip

It remains to prove Proposition~\ref{prop:covesynch}. We first define an induction parameter. We say that $(q,a,r) \in \delta$ is a \emph{frontier transition} if the states $q$ and $r$ are \emph{not} strongly connected. Moreover, given $q,r\in Q$ and $w \in A^*$, we associate a number $d(q,w,r) \in \nat \cup \{\infty\}$. If $w \not\in \alauto{q}{r}$, we let $d(q,w,r) = \infty$. Otherwise, $w \in \alauto{q}{r}$ and $d(q,w,r)$ is the \emph{least} number $n \in \nat$ such that there is a run from $q$ to $r$ labeled by $w$ in \As using \emph{exactly $n$ frontier transitions}. Note that $d(q,w,r) = 0$ if and only if $w \in \alauto{q}{r}$  and $q,r$ are strongly connected. One may verify the following fact.

\begin{fact}\label{fct:dist}
  Let $q,r\in Q$ and $w \in A^*$ such that $w \in \alauto{q}{r}$. Then, $d(q,w,r) \leq |Q|-1$. Also, for all $u,v \in A^*$ if $w=uv$, there is $s \in Q$ such that $d(q,u,s) + d(s,v,r) = d(q,w,r)$.
\end{fact}

Let $\ell \in \nat$. An \emph{$\ell$-synchronizer} is a morphism $\alpha: A^* \to G$ into a finite group $G$ satisfying the two following properties:
\begin{enumerate}
  \item for all $q,r \in Q$ and $w \in A^*$ such that $d(q,w,r) \leq \ell$ and $\alpha(w) = 1_G$, we have $\veps \in \aelauto{q}{r}$.
  \item for all $n \geq 1$, all $q_1,\dots,q_n,r_1,\dots,r_n \in Q$ and all $w_1,\dots,w_n \in A^*$ such that $\sum_{i \leq n} d(q_i,w_i,r_i) \leq \ell-1$ and $\alpha(w_1) = \cdots = \alpha(w_n)$, there exists $u \in \tilas$ such that $u \in \aelauto{q_i}{r_i}$ for every $i \leq n$.
\end{enumerate}

\begin{remark}\label{rem:grp:second}
  There is a subtle difference between Properties~1 and 2. The first requires that $d(q,w,r) \leq \ell$ while the second requires that $\sum_{i \leq k} d(q_i,w_i,r_i) \leq \ell-1$. In particular, when $\ell = 0$, the second property is trivially satisfied since $\sum_{i \leq k} d(q_i,w_i,r_i)$ cannot be smaller than $-1$.
\end{remark}

We first show that thanks to Property~2, any $\ell$-synchronizer for $\ell$ large enough is also a synchronizer (on the other hand, we do not need Property~1 at this stage).

\begin{lemma}\label{lem:lsynchtosynch}
  Let $\As=(Q,\delta,I,F)$ be an \nfa and $\ell = |Q|^3$. Every $\ell$-synchronizer is also a synchronizer.
\end{lemma}

\begin{proof}
  Let $\alpha: A^* \to G$ be an $\ell$-synchronizer. We show that it is a synchronizer. Let $q_1,\dots,q_n, r_1,\dots,r_n \in Q$ and let $g \in G$ such that $\alpha\inv(g) \cap \alauto{q_j}{r_j} \neq \emptyset$ for all $j\leq n$.  We exhibit $u \in \tilas$ such that $u \in \aelauto{q_j}{r_j}$ for all~$j \leq n$. Clearly, we may assume without loss of generality that for all $i,j \leq n$ such that $i \neq j$, we have $(q_i,r_i)\neq(q_j,r_j)$. This implies that $n \leq |Q|^2$. Our hypothesis yields $w_j \in  \alpha\inv(g) \cap \alauto{q_j}{r_j}$ for every $j \leq n$. By Fact~\ref{fct:dist}, we have $d(q_j,w_j,r_j) < |Q|$. Since $n \leq |Q|^2$, it follows that $\sum_{j \leq n} d(q_j,w_j,r_j)< |Q|^3= \ell$. Moreover, we have $\alpha(w_1) = \cdots = \alpha(w_n) = g$ by definition. Hence, since $\alpha$ is an $\ell$-synchronizer, we get $u \in \tilas$ such that $u \in \aelauto{q_i}{r_i}$ for every $i \leq n$, as desired.
\end{proof}

In view of Lemma~\ref{lem:lsynchtosynch}, it suffices to prove that for each $\ell \in \nat$, there exists an $\ell$-synchronizer. Indeed, the case when $\ell = |Q|^3$ yields the synchronizer described in Proposition~\ref{prop:covesynch}.

We shall use induction on $\ell \in \nat$ to build an $\ell$-synchronizer (it is for this induction that Property~1 will be useful). We devote the remainder of the section to this proof. Before starting the induction, we state simple lemmas on~$\ell$-synchronizers.

\smallskip
\noindent
{\bf Preliminaries.} For each $q \in Q$, we define a set $L(q) \subseteq \tilas$. For $v \in \tilas$, we let $v \in L(q)$ if and only if there exists $q' \in Q$ such that $q,q'$ are strongly connected and $v \in \aelauto{q}{q'}$.

\begin{lemma}\label{lem:thekey}
  Let $\ell\in\nat$, let $\beta:A^*\to H$ be an $\ell$-synchronizer, let $s,t \in Q$ and let $w \in A^*$ such that $d(s,w,t) \leq \ell$. Then:
  \begin{itemize}
    \item \!If $v \in L(s)$ and $\beta(w) =\beta(v)$, then $v \in \aelauto{s}{t}$.
    \item \!If $v \in  L(t)$ and $\beta(w) = \beta(v)\inv$, then \mbox{$v\inv \in \aelauto{s}{t}$}.
  \end{itemize}
\end{lemma}

\begin{proof}
  For the first assertion, consider $v \in L(s)$ such that $\beta(w) = \beta(v)$. By definition of $L(s)$, we get $s' \in Q$ such that $s,s'$ are strongly connected and $v \in \aelauto{s}{s'}$. Thus, we get $v\inv \in \aelauto{s'}{s}$ by Fact~\ref{fct:jump} and Lemma~\ref{lem:jtrans} yields $x \in \alauto{s'}{s}$ such that $\beta(x) = \beta(v\inv)$. Since $d(s,w,t) \leq \ell$ and $s',s$ are strongly connected, it follows that $d(s',xw,t) \leq \ell$. Moreover, since $\beta(w) = \beta(v)$ and $\beta(x) = \beta(v\inv)$, we have $\beta(xw) = 1_H$. Altogether, since $\beta$ is an $\ell$-synchronizer, we get $\veps \in \aelauto{s'}{t}$ by Property~1. Since $v \in \aelauto{s}{s'}$, we get $v \in \aelauto{s}{t}$ as desired.

  For the second assertion, let $v \in L(t)$ and assume that $\beta(w) = (\beta(v))\inv$. By definition of $L(t)$, we have $t' \in Q$ such that $t,t'$ are strongly connected and $v \in \aelauto{t}{t'}$. Lemma~\ref{lem:jtrans} yields $y \in \alauto{t}{t'}$ such that $\beta(y) = \beta(v)$. Since $d(s,w,t) \leq \ell$ and $t,t'$ are strongly connected, we get $d(s,wy,t') \leq \ell$. Also, $\beta(w)\! =\! (\beta(v))\inv$ and $\beta(y)\! =\! \beta(v)$. Thus, $\beta(wy) = 1_H$ and since $\beta$ is an $\ell$-synchronizer, we get $\veps \in \aelauto{s}{t'}$. Finally, since $v \in \aelauto{t}{t'}$ and $t,t'$ are strongly connected, we get $v\inv \in \aelauto{t'}{t}$ by Fact~\ref{fct:jump}. Altogether, we obtain  $v\inv \in \aelauto{s}{t}$.
\end{proof}

\noindent
Let $\beta:A^*\to H$ be an $\ell$-synchronizer. Consider $q \in Q$ and $(h,a) \in H \times A$. We say that $(h,a)$ \emph{stabilizes}~$q$ if there are $x \in A^*$ and $s \in Q$ such that $d(q,xa,s) = 0$ and $\beta(x)=h$. The next lemma follows from Lemma~\ref{lem:thekey}.

\begin{lemma}\label{lem:stabalt}
  Let $\ell\in\nat$ and $\beta:A^*\to H$ be an $\ell$-synchronizer. Let $q \in Q$ and $(h,a) \in H \times A$ that stabilizes $q$.~Then:
  \begin{itemize}
    \item If $v \in L(q)$ and $\beta(v) = h$, then $va \in L(q)$.
    \item If $v \in L(q)$ and $\beta(v) = h\beta(a)$, then $va\inv \in L(q)$.
  \end{itemize}
\end{lemma}

\begin{proof}
  Since $(h,a)$ stabilizes $q$, we have $x \in A^*$ and $s \in Q$ such that $d(q,xa,s) = 0$ and $\beta(x)=h$. Since $d(q,xa,s) = 0$, Fact~\ref{fct:dist} yields $q' \in Q$ such that $d(q,x,q') = d(q',a,s) = 0$.

  Let $v \in L(q)$ such that $\beta(v) = h=\beta(x)$. Lemma~\ref{lem:thekey} yields $v \in \aelauto{q}{q'}$ since $d(q,x,q') = 0\leq\ell$. Since \mbox{$(q',a,s) \in \delta$}, we get $va \in \aelauto{q}{s}$. This yields $va \in L(q)$ since $q,s$ are strongly connected.

  We now consider $v \in L(q)$ such that $\beta(v) = h\beta(a) = \beta(xa)$. Since $d(q,xa,s) = 0\leq\ell$, Lemma~\ref{lem:thekey} yields $v \in\aelauto{q}{s}$. Moreover, since $(q',a,s) \in \delta$ and $q',s$ are strongly connected, we have $(s,a\inv,q') \in \acdel$ by definition. Thus, $va\inv\in \aelauto{q}{q'}$. Since $q,q'$ are strongly connected, this yields $va\inv \in L(q)$, as~desired.
\end{proof}

\smallskip\noindent
\textbf{Construction of $\ell$-synchronizers by induction on~$\ell$.} We are ready to prove that for all $\ell\in\nat$, there exists an $\ell$-synchronizer.

\smallskip
\noindent
{\bf Base case: $\ell = 0$.} The definition of our $0$-synchronizer is based on an equivalence. Let $q,r \in Q$. We write $q \simeq r$ when $q$ and~$r$ are strongly connected and $\veps \in \aelauto{q}{r}$.

\begin{lemma}\label{lem:grp:bsync}
  The relation $\simeq$ is an equivalence. Moreover, for every $q,r,q',r' \in Q$ which are strongly connected  and $a \in A$, if $(q,a,q') \in \delta$ and $(r,a,r') \in \delta$, then $q \simeq r \Leftrightarrow q' \simeq r'$.
\end{lemma}

\begin{proof}
  Clearly, $\simeq$ is reflexive: $\veps \in \aelauto{q}{q}$ for every $q\in Q$. Moreover, if $q \simeq r$, then $q$ and $r$ are strongly connected and $\veps \in \aelauto{q}{r}$. Consequently, since $\veps= \veps\inv$, Fact~\ref{fct:jump} yields $\veps \in \aelauto{r}{q}$ and we get  $r \simeq q$. Hence $\simeq$ is symmetric. Finally, let $q,r,s \in Q$ such that $q \simeq r$ and $r \simeq s$. By definition, $q,r,s$ are strongly connected, $\veps \in \aelauto{q}{r}$ and $\veps \in \aelauto{r}{s}$. Clearly, $\veps \in \aelauto{q}{s}$ which yields $q \simeq s$ and we conclude that $\simeq$ is transitive.

  We now prove that for all $q,r,q',r' \in Q$ which are strongly connected and $a \in A$ such that $(q,a,q') \in \delta$ and $(r,a,r') \in \delta$, we have $q \simeq r \Leftrightarrow q' \simeq r'$. By definition, $(q',a\inv,q) \in \acdel$ and $(r',a\inv,r) \in \acdel$.  Assume first that $q \simeq r$. We have $\veps \in \aelauto{q}{r}$. Thus, $a\inv a \in \aelauto{q'}{r'}$ and since $a\inv a \xrightarrow{*} \veps$, Fact~\ref{fct:jump} yields $\veps \in \aelauto{q'}{r'}$. Hence, $q' \simeq r'$. Conversely, if $q' \simeq r'$, we have $\veps \in \aelauto{q'}{r'}$. Thus, $aa\inv \in \aelauto{q}{r}$ and since $aa\inv \xrightarrow{*} \veps$, Fact~\ref{fct:jump} yields $\veps \in \aelauto{q}{r}$. We get $q \simeq r$, as desired.
\end{proof}

For each $q \in Q$, we write $\folcl{q} \in {Q}/{\simeq}$ for the $\simeq$-class of~$q$. Moreover, we let $G$ be the group of permutations of ${Q}/{\simeq}$. That is, $G$ consists of all bijections $g: {Q}/{\simeq}\to {Q}/{\simeq}$ and the multiplication is composition (the neutral element is identity). We have the following fact.

\begin{fact}\label{fct:grp:defm}
  For every $a \in A$, there exists an element $g_a \in G$ such that for every $q,q' \in Q$ which are strongly connected and such that $(q,a,q') \in \delta$, we have $g_a(\folcl{q}) = \folcl{q'}$.
\end{fact}

\begin{proof}
  Consider $q \in Q$. By Lemma~\ref{lem:grp:bsync}, if there exists $q'\in Q$ such that $q,q'$ are strongly connected and $(q,a,q') \in \delta$, we know that for every $r,r' \in Q$ which are strongly connected and such that $(r,a,r') \in \delta$, we have $q \simeq r \Leftrightarrow q' \simeq r'$. Hence, we may define $g_a(\folcl{q}) = \folcl{q'}$. This yields a \emph{partial} function $g_a: {Q}/{\simeq} \to {Q}/{\simeq}$ which satisfies the condition described in the fact and is injective. Hence, we may complete $g_a$ into a bijection, concluding the proof.
\end{proof}

We let $\alpha: A^* \to G$ be the morphism defined by $\alpha(a)=g_a$ for every $a \in A$ and show that $\alpha$ is a $0$-synchronizer. We prove the first property in the definition (the second one is trivially satisfied when $\ell = 0$). Let $q,r \in Q$ and $w \in A^*$, such that $d(q,w,r) = 0$ and $\alpha(w) = 1_G$. By definition,  $\alpha(w)$ is a permutation of ${Q}/{\simeq}$. Moreover, since $d(q,w,r) = 0$, we have $w \in \alauto{q}{r}$ and $q,r$ are strongly connected. By definition of $\alpha$ from Fact~\ref{fct:grp:defm} this implies that $\alpha(w)(\folcl{q}) = \folcl{r}$. Finally, since $\alpha(w) = 1_G$, we also have $\alpha(w)(\folcl{q}) = \folcl{q}$. Hence, $q \simeq r$  and the definition of $\simeq$ yields $\veps \in \aelauto{q}{r}$. We conclude that $\alpha$ is a $0$-synchronizer.

\smallskip
\noindent
{\bf Inductive step: $\ell \geq 1$.} By induction on $\ell$, we know that there exists an $(\ell - 1)$-synchronizer $\beta: A^* \to H$. We use it to construct a new morphism $\alpha: A^* \to G$ from $\beta$. Then, we prove that $\alpha$ is an $\ell$-synchronizer.

For every pair $(h,a) \in H \times A$ and every $w \in A^*$, we let  $\#_{h,a}(w) \in \nat$ be the number of pairs $(x,y) \in A^* \times A^*$ such that $\beta(x) = h$  and $w = xay$. The definition of the morphism $\alpha: A^* \to G$ is designed with the following goal in mind: for each word $w \in A^*$, we want its image \mbox{$\alpha(w) \in G$} to determine $\beta(w)\in H$ and, for every $(h,a) \in H \times A$, whether the number $\#_{h,a}(w) \in \nat$ is even or odd. The definition is inspired by the work of Auinger~\cite{Auinger04}. We let $G = H \times \{0,1\}^{H \times A}$. That is, every element $g\in G$ is a pair $g = (h,f)$ where $h \in H$ and \mbox{$f: H \times A \to \{0,1\}$} is a function. We now equip~$G$ with a multiplication. Let $g_1,g_2\in G$ with $g_1 = (h_1,f_1)$ and $g_2= (h_2,f_2)$. We define $g_1g_2= (h_1h_2,f)$ where $f: H \times A \to  \{0,1\}$ is the function $f: (h,a) \mapsto f_1(h,a) + f_2(h_1\inv h,a) \bmod 2$. One may verify that $G$ is indeed a group for this multiplication (technically, $G$ is a wreath product, see \emph{e.g.},~\cite{Almeida:1994a}). For every $w \in A^*$, let $f_w: H \times A \to \{0,1\}$ be the function defined by $f_w(h,a) = \#_{h,a}(w) \bmod 2$. One may now verify that the map $\alpha: A^* \to G$ defined by $\alpha(w) = (\beta(w),f_w)$ is a monoid morphism. It remains to show that it is an $\ell$-synchronizer.

We first explain how to exploit the definition of $\alpha$. A key point is that we are interested in special pairs $(h,a) \in H \times A$. Given $F \subseteq H$, we say that such a pair $(h,a)$ is \emph{$F$-alternating} when $h\in F \Leftrightarrow h\beta(a)\not\in F$. Moreover, we say that a word $w \in A^*$ is \emph{$F$-safe} if $\#_{h,a}(w)$ is \emph{even} for every $F$-alternating pair $(h,a) \in H \times A$. By definition, the image $\alpha(w) \in G$ determines whether $w$ is $F$-safe or not. In the latter case, we get an $F$-alternating pair $(h,a)$ such that $\#_{h,a}(w)$ is \emph{odd} (and thus, $\#_{h,a}(w)\geq 1$). In the former, we use the next lemma.

\begin{lemma}\label{lem:alt}
  Let $F \subseteq H$ such that $1_H \in F$. For every $w \in A^*$ which is $F$-safe, $\beta(w) \in F$.
\end{lemma}

\begin{proof}
  We prove a stronger property. For every $w \in A^*$, we write $\#_F(w) \in \nat$ for the sum of all numbers $\#_{h,a}(w)$ where $(h,a) \in H\times A$ is $F$-alternating. We prove that for $w \in A^*$, we have $\beta(w) \in F \Leftrightarrow \#_F(w)$ \emph{is even}. This implies the lemma: if $w$ is $F$-safe, then  $\#_F(w)$ is even which yields $\beta(w) \in F$.

  We use induction on the length of $w\in A^*$. If $w = \veps$, then $\beta(w) = 1_H \in F$ and $\#_F(w)= 0$. Thus, the property is trivially satisfied. Assume now that $w \in A^+$. This yields $v\in A^*$ and $a\in A$ such that $w=va$. Clearly, $|v| < |w|$ which yields $\beta(v) \in F \Leftrightarrow \#_F(v)$ \emph{is even} by induction. It follows that $\beta(v) \not\in F \Leftrightarrow \#_F(v)$ \emph{is odd}. There are two cases. First, assume that $(\beta(v),a)$ is $F$-alternating. In that case, since $w = va$, it follows that $\beta(w) \in F \Leftrightarrow \beta(v) \not\in F$ and $\#_F(w) = \#_F(v) +1$ (\emph{i.e.}, $\#_F(w)$ \emph{is even} $\Leftrightarrow$ $\#_F(v)$ \emph{is odd}). Thus, we may combine the equivalences to get $\beta(w) \in F \Leftrightarrow \#_F(w)$ \emph{is even} as desired. Assume now that $(\beta(v),a)$ is \emph{not} $F$-alternating. In that case, as $w = va$, we get $\beta(w) \in F \Leftrightarrow \beta(v) \in F$ and $\#_F(w) = \#_F(v)$ (thus, $\#_F(w)$ \emph{is even} $\Leftrightarrow$ $\#_F(v)$ \emph{is even}). Hence, we may again combine the equivalences to get $\beta(w) \in F \Leftrightarrow \#_F(w)$ \emph{is even}, as desired.
\end{proof}

We now present the sets $F \subseteq H$ to be used in Lemma~\ref{lem:alt}.  Recall that for each $q \in Q$, the language $L(q) \subseteq \tilas$ consists of all words $v \in \tilas$ satisfying $v \in \aelauto{q}{q'}$ for some $q' \in Q$ such that $q,q'$ are \emph{strongly connected}. To each $S \subseteq Q$, we associate a set $F_S \subseteq H$ as follows,
\[
  F_S = \Biggr\{\beta(v) \mid v \in \bigcap_{q \in S} L(q)\Biggr\}.
\]
A key point is that $1_H \in F_S$ for all $S \subseteq Q$. Indeed, $\veps \in L(q)$ for all $q \in Q$ since $\veps \in \aelauto{q}{q}$. Hence, Lemma~\ref{lem:alt} applies to $F_S$. Finally, we present a corollary of Lemma~\ref{lem:stabalt}. Recall that $(h,a)\in H\times A$ \emph{stabilizes}~$q$ if there exist $x,y \in A^*$ and $s \in Q$ such that $d(q,xay,s) = 0$ and $\beta(x)=h$.

\begin{corollary}\label{cor:mainstuff}
  If $S \subseteq Q$ and $(h,a) \in H \times A$ is $F_S$-alternating, there exists $q \in S$ such that $(h,a)$ does not stabilize $q$.
\end{corollary}

\begin{proof}
  By contradiction, assume that $(h,a)$ stabilizes $q$ for all $q \in S$. We show that $h \in F_S \Leftrightarrow h\beta(a) \in F_S$, contradicting the hypothesis that $(h,a)$ is $F_S$-alternating. Assume first that $h \in F_S$. By definition, this yields $v \in \bigcap_{q \in S} L(q)$ such that $\beta(v) = h$. As $(h,a)$ stabilizes $q$ for all $q \in S$, the first assertion in Lemma~\ref{lem:stabalt} yields $va \in \bigcap_{q \in S} L(q)$. Thus, $h\beta(a)\in F_S$. Conversely assume that $h\beta(a)\in F_S$. By definition, this yields $v' \in \tilas$ such that $v' \in L(q)$ for all $q \in S$ and $\beta(v') =h\beta(a)$. Since $(h,a)$ stabilizes $q$ for all $q \in S$, the second assertion in Lemma~\ref{lem:stabalt} yields $v'a\inv \in \bigcap_{q \in S} L(q)$. Thus, $h \in F_S$, as~desired.
\end{proof}

We are ready to prove that $\alpha$ is an $\ell$-synchronizer. There are two conditions to prove.

\smallskip
\noindent
{\bf Condition 1.} Let $q,r\! \in\! Q$ and $w\! \in\! A^*$ such that $d(q,w,r)\! \leq\! \ell$ and $\alpha(w)=1_G$. We show that $\veps \in \aelauto{q}{r}$. We have $\beta(w) = 1_H$ by definition of $\alpha$. Hence, since $\beta$ is an $(\ell-1)$-synchronizer, the result is immediate when $d(q,w,r) \leq \ell-1$. We assume from now on that $d(q,w,r)= \ell$.

Since $\ell \geq 1$, $w$ is nonempty. Let $a_1, \dots ,a_n \in A$ such that $w = a_1 \cdots a_n$. By Fact~\ref{fct:dist}, we have $q_0,\dots,q_n \in Q$ such that $q_0 = q$, $q_n = r$ and $\sum_{1\leq k \leq n} d(q_{k-1},a_{k},q_{k}) = d(q,w,r) = \ell$. This means that there are exactly $\ell$ indices $k < n$ such that $(q_{k-1},a_{k},q_{k}) \in \delta$ is a frontier transition. For $0 \leq k \leq n$, we let $x_k = a_1 \cdots a_k$ and $y_k = a_{k+1} \cdots a_n$ (we let $x_0 = y_n = \veps$). Clearly, $w= x_ky_k$. We let $h_k = \beta(x_k)$ for every $k \leq n$. Note that since $\beta(w) = 1_H$, we also know that $\beta(y_k) = h_k\inv$.

Let $i \leq n$ be the least index such that $(q_{i-1},a_{i},q_{i})$ is a frontier transition. Let $j\leq  n$ be the greatest index such that $(q_{j-1},a_{j},q_{j})$ is a frontier transition. Clearly, $1 \leq i \leq j \leq n$ ($i = j$, if $\ell = 1$). By definition, we have the following fact.

\begin{fact}\label{fct:grp:middle}
  Let $k \leq n$. If $i \leq k$, then $d(q_k,y_k,r) \leq \ell-1$. If $k < j$, then $d(q,x_k,q_k) \leq \ell-1$.
\end{fact}

The hypothesis that $\alpha(w) = 1_G$ implies the next lemma.

\begin{lemma}\label{lem:grp:cond1}
  One of the three following properties holds:
  \begin{enumerate}
    \item there exists $k$ such that $i \leq k < j$ and $h_k \in \setfqr$, or,
    \item $h_{i-1}  \in \setfqr$ and $(h_{i-1} ,a_{i})$ stabilizes $r$, or,
    \item $h_{j} \in \setfqr$ and $(h_{j-1},a_{j})$ stabilizes $q$.
  \end{enumerate}
\end{lemma}

\begin{proof}
  Since $\alpha(w) = 1_G$ and $w = x_{j}y_{j}$,  it follows that $\alpha(x_{j}) = \alpha(y_{j}\inv)$. Moreover, $y_j \in \alauto{q_j}{r}$. Consequently, $y_j\inv \in \alauto{r}{q_j}$ since $q_{j},r$ are strongly connected by definition of $j$. Thus, Lemma~\ref{lem:trans} yields $z \in \alauto{q_j}{r}$ such that $\alpha(z) = \alpha(y_{j}\inv) = \alpha(x_{j})$. For all $(h,a) \in H \times A$, we have the following two properties:
  \begin{itemize}
    \item By definition of $i$, we have $d(q,x_{i-1},q_{i-1}) = 0$. Thus,  if $\#_{h,a}(x_{i-1}) \geq 1$, then $(h,a)$ stabilizes $q$.
    \item By definition of $j$, we have $d(r,z,q_{j}) = 0$. Thus, if $\#_{h,a}(z)\geq 1$, then $(h,a)$ stabilizes $r$.
  \end{itemize}
  We use these properties and their contrapositives repeatedly. We consider two cases depending on whether $x_i$ is \setfqr-safe.

  \smallskip
  \noindent
  {\it Case~1: $x_i$ is \setfqr-safe.} We know that $h_i = \beta(x_i) \in \setfqr$ by Lemma~\ref{lem:alt}. Clearly, if $i < j$, then Assertion~1 in the lemma holds for $k = i$ and we are finished. Assume now that $i = j$. Let $(h,a) =(h_{i-1},a_i)=(h_{j-1},a_j)$.  The argument depends on whether $\#_{h,a}(x_{i-1}) \geq 1$ or not. If $\#_{h,a}(x_{i-1}) \geq 1$, then $(h,a) = (h_{j-1},a_j)$ stabilizes~$q$. Thus, Assertion~3 in the lemma holds as $h_j \!=\! h_i\! \in\! \setfqr$. Otherwise, $\#_{h,a}(x_{i-1}) = 0$. Since $x_i = x_{i-1}a_i$ and $(h,a) =(h_{i-1},a_i)$, it follows that $\#_{h,a}(x_{i}) = 1$. Thus, $\#_{h,a}(x_{i})$ is odd and since $x_i$ is \setfqr-safe, it follows that $(h,a)$ is \emph{not} \setfqr-alternating. Since $h\beta(a) = h_i \in \setfqr$, we also have $h_{i-1} = h \in \setfqr$. Finally, as $x_i = x_j$, we have $\alpha(x_i) = \alpha(x_{j}) = \alpha(z)$. Thus, as $\#_{h,a}(x_{i}) = 1$, we get that $\#_{h,a}(z)$ is odd by definition of $\alpha$. Hence, $\#_{h,a}(z) \geq 1$ which yields that $(h_{i-1},a_i) = (h,a)$ stabilizes~$r$. As $h_{i-1}\in \setfqr$, it follows that  Assertion~2~holds.

  \smallskip
  \noindent
  {\it Case~2: $x_i$ is not \setfqr-safe.} The argument depends on whether $x_{i-1}$ is \setfqr-safe or not. Assume first that $x_{i-1}$ is \setfqr-safe. Lemma~\ref{lem:alt} yields $h_{i-1} = \beta(x_{i-1}) \in \setfqr$. If there exists $k$ such that $i \leq k < j$ and $h_k = h_{i-1}$, then Assertion~1 in the lemma holds. Otherwise, we have $\#_{h_{i-1},a_i}(x_{i}) = \#_{h_{i-1},a_i}(x_{j})$. By hypothesis, $x_i = x_{i-1}a_i$ is not \setfqr-safe while $x_{i-1}$ is \setfqr-safe. Thus, $(h_{i-1},a_i)$ is \setfqr-alternating and  $\#_{h_{i-1},a_i}(x_{i}) = \#_{h_{i-1},a_i}(x_{j})$ is odd. Since $\alpha(z)=\alpha(x_{j})$, it follows that $\#_{h_{i-1},a_i}(z)$ is also odd by definition of $\alpha$. Thus, $\#_{h_{i-1},a_i}(z) \geq 1$ which implies that $(h_{i-1},a_i)$ stabilizes~$r$. Since $h_{i-1} \in \setfqr$,  Assertion~2 in the lemma holds.

  Finally, assume that $x_{i-1}$ is not \setfqr-safe: we have $(h,a)$ which is \setfqr-alternating and such that $\#_{h,a}(x_{i-1})$ is odd. Since $x_i = x_{i-1}a_i$ is not \setfqr-safe as well, we may choose $(h,a)$ so that $(h,a) \neq (h_{i-1},a_i)$. Thus, $\#_{h,a}(x_{i})$ is odd as well. Since $\#_{h,a}(x_{i-1}) \geq 1$, we know that $(h,a)$ stabilizes $q$. By Corollary~\ref{cor:mainstuff} it follows that $(h,a)$ does \emph{not} stabilize~$r$. This implies $\#_{h,a}(z)= 0$ and since $\alpha(z)=\alpha(x_{j})$, it follows that $\#_{h,a}(x_{j})$ is even. Since $\#_{h,a}(x_{i})$ is odd, this yields $k$ such that $i \leq k < j$ and $(h_k,a_{k+1}) = (h,a)$. Since $(h,a)$ is \setfqr-alternating either $h_k \in \setfqr$ or $h_{k+1} =h_k\beta(a_{k+1}) \in \setfqr$. If $h_k \in \setfqr$, Assertion~1 holds. If $h_{k+1} \in \setfqr$, then either $i \leq k < j-1$ and Assertion~1 in the lemma holds, or $k = j-1$ which means that $h_j = h_{k+1} \in \setfqr$ and $(h_{j-1},a_{j}) = (h,a)$ which stabilizes $q$: Assertion~3 in the lemma holds.
\end{proof}

We may now prove that $\veps \in \aelauto{q}{r}$. We treat the three cases depicted in Lemma~\ref{lem:grp:cond1} independently. First, assume that there exists $k$ such that $i \leq k < j$ and $h_k \in \setfqr$. The definition of \setfqr yields $v \in L(q) \cap L(r)$ such that $\beta(v) = h_k$. It follows from Fact~\ref{fct:grp:middle} that $d(q,x_k,q_k) \leq \ell-1$. Therefore, since $v \in L(q)$ and $\beta(x_k) = h_k= \beta(v)$, Lemma~\ref{lem:thekey} implies that $v \in \aelauto{q}{q_k}$. Symmetrically, Fact~\ref{fct:grp:middle} yields $d(q_k,y_k,r) \leq \ell-1$. Thus, since we have $v \in L(r)$ and $\beta(y_k) = h_k\inv = (\beta(v))\inv$, it follows from Lemma~\ref{lem:thekey} that $v\inv \in \aelauto{q_k}{r}$. Hence, $vv\inv \in \aelauto{q}{r}$. Since $vv\inv \xrightarrow{*} \veps$, Fact~\ref{fct:jump} yields $(q,\veps,r) \in \ecdel$ concluding this case.

In the second case, $h_{i-1}\!\in\! \setfqr$ and $(h_{i-1} ,a_{i})$ stabilizes~$r$. By definition of \setfqr, we have $v \in L(q) \cap L(r)$ such that $\beta(v) = h_{i-1}$. We have $d(q,x_{i-1},q_{i-1}) = 0$ by definition of~$i$. Thus, as $v \in L(q)$ and $\beta(x_{i-1}) = \beta(v)$, Lemma~\ref{lem:thekey} yields $v \in \aelauto{q}{q_{i-1}}$. Moreover, since $(h_{i-1},a_{i})$ stabilizes $r$, $v \in L(r)$ and $\beta(v) = h_{i-1}$, Lemma~\ref{lem:stabalt} implies that $va_i \in L(r)$. Fact~\ref{fct:grp:middle} yields $d(q_{i},y_i,r) \leq \ell-1$. Thus, since $va_i \in L(r)$ and $\beta(y_i) = h_i\inv = (\beta(va_i))\inv$, Lemma~\ref{lem:thekey} yields $(va_i)\inv \in \aelauto{q_i}{r}$. Hence, since $(q_{i-1},a_{i},q_i) \in \delta$, we get $va_i(va_i)\inv\in\aelauto{q}{r}$. Since $va_i(va_i)\inv \xrightarrow{*} \veps$, it follows that $(q,\veps,r) \in \ecdel$ by Fact~\ref{fct:jump}, concluding this case.

In the last case, $h_{j} \in \setfqr$ and $(h_{j-1},a_{j})$ stabilizes $q$. By definition, we get $v \in L(q) \cap L(r)$ such that $\beta(v) = h_{j}$. We have $d(q_j,y_j,r) = 0$ by definition of $j$. As $v \in L(r)$ and $\beta(y_j) \!=\! h_j\inv\! =\!\beta(v)\inv$, we get $v\inv \in \aelauto{q_j}{r}$ by Lemma~\ref{lem:thekey}. Moreover, we know that $(h_{j-1},a_{j})$ stabilizes $q$, $v \in L(q)$ and $\beta(v) = h_j = h_{j-1}\beta(a_j)$. Thus, $va_j\inv \in L(q)$ by Lemma~\ref{lem:stabalt}. We have $d(q,x_{j-1},q_{j-1}) \leq \ell-1$ by Fact~\ref{fct:grp:middle}. Thus, since $va_j\inv \in L(q)$ and $\beta(x_{j-1}) = h_{j-1} = \beta(va_j\inv)$, it follows from Lemma~\ref{lem:thekey} that $va_j\inv \in \aelauto{q}{q_{j-1}}$. Since $(q_{j-1},a_{j},q_j) \in \delta$, we obtain $va_j\inv a_jv\inv \in \aelauto{q}{r}$. Thus, since $va_j\inv a_jv\inv \xrightarrow{\smash{*}} \veps$, Fact~\ref{fct:jump} yields $(q,\veps,r) \in \ecdel$ as desired. This concludes the proof for the first condition.

\smallskip
\noindent
{\bf Condition~2.} Consider $q_1, \dots ,q_n\in Q$, $r_1, \dots ,r_n  \in  Q$ and $w_1, \dots, w_n \in  A^*$ such that $\sum_{i \leq n} d(q_i,w_i,r_i)  \leq \ell-1$ and $\alpha(w_1) = \cdots = \alpha(w_n)$. We need to exhibit $u \in \tilas$ such that $u \in \aelauto{q_i}{r_i}$ for every $i \leq k$.  By definition of $\alpha$, we have $\beta(w_1) = \cdots = \beta(w_n)$. Let $S = \{q_1,\dots,q_n\}$. There are two cases depending on whether $w_1$ is $F_S$-safe or not.

Assume first that $w_1$ is $F_S$-safe. By Lemma~\ref{lem:alt}, it follows that $\beta(w_1) \in F_S$. We get $u \in \tilas$ such that $u \in \bigcap_{i \leq n} L(q_i)$ and $\beta(u) = \beta(w_1) = \cdots = \beta(w_n)$. Since $d(q_i,w_i,r_i) \leq \ell-1$ by hypothesis, Lemma~\ref{lem:thekey} yields $u \in \aelauto{q_i}{r_i}$  for every $i \leq k$, concluding this case.

Conversely, we assume that $w_1$ is \emph{not} $F_S$-safe. By definition, this yields an $F_S$-alternating pair $(h,a)$ such~that $\#_{h,a}(w_1)$ is odd. By definition of $\alpha$, it follows that~$\#_{h,a}(w_i)$ is odd as well for every index $i \leq n$ since \mbox{$\alpha(w_1) = \alpha(w_i)$}. Therefore, we have $\#_{h,a}(w_i) \geq 1$ for every $i \leq n$. We get $x_i,y_i\! \in\! A^*$ such that $w_i = x_iay_i$ and \mbox{$\beta(x_i) = h$}. Since \mbox{$d(q_i,w_i,r_i) \in \nat$}, we get \mbox{$d(q_i,x_i,s_i) + d(s_i,a,t_i) + d(t_i,y_i,r_i) = d(q_i,w_i,r_i)$} for $s_i,t_i \in Q$ by Fact~\ref{fct:dist}. Thus, $d(t_i,y_i,r_i) \leq d(q_i,w_i,r_i)$ for all $i \leq k$. Also, $(h,a)$ is $F_S$-alternating and \mbox{$S = \{q_1,\dots,q_n\}$}. Hence, Corollary~\ref{cor:mainstuff} yields $j \leq k$ such that $(h,a)$ does not stabilize~$q_j$.  As \mbox{$\#_{h,a}(x_ja) \geq 1$}, we get $d(q_j,x_ja_j,t_j) \geq 1$ which yields the  \emph{strict} inequality $d(t_j,y_j,r_j) < d(q_j,w_j,r_j)$. Altogether, we obtain $\sum_{i \leq k} d(t_i,y_i,r_i) < \sum_{i \leq k} d(q_i,w_i,r_i)$. By hypothesis, this implies that $\sum_{i \leq k} d(t_i,y_i,r_i) \leq (\ell-1) - 1$. Moreover, $H$ is a group, $\beta(x_1a) = \cdots = \beta(x_na) = h\beta(a)$ and $\beta(w_1) = \cdots = \beta(w_n)$. Hence, $\beta(y_1) = \cdots = \beta(y_n)$ and since $\beta$ is an $(\ell-1)$-synchronizer, we obtain $z \in \tilas$ such that $z \in \aelauto{t_i}{r_i}$ for every $i \leq k$.

We now consider two subcases. Since the pair $(h,a)$ is $F_S$-alternating, either $h \in F_S$ or $h\beta(a) \in F_S$. If $h \in F_S$, we get a word $v\! \in\! \bigcap_{i \leq k}\! L(q_i)$ such that $\beta(v)\! =\! h\! =\!\beta(x_i)$ for all \mbox{$i\! \leq\! n$}. Thus, since $d(q_i,x_i,s_i) \leq d(q_i,w_i,r_i) \leq \ell-1$, Lemma~\ref{lem:thekey} yields $v \in \aelauto{q_i}{s_i}$ for all $i \leq k$. Moreover, we have $(s_i,a,t_i) \in \delta$. Altogether, it follows that $vaz \in \aelauto{q_i}{r_i}$ for every $i \leq k$. This concludes the first subcase. Finally, assume that $h\beta(a) \in F_S$. This yields $v' \in \bigcap_{i \leq k} L(q_i)$ such that  $\beta(v') = h\beta(a) = \beta(x_ia)$ for all $i \leq n$. Since we know that $d(q_i,x_ia_i,t_i) \leq d(q_i,w_i,r_i) \leq \ell-1$,  Lemma~\ref{lem:thekey} yields $v' \in \aelauto{q_i}{t_i}$ for every $i \leq k$. Altogether, it follows that $v'z \in \aelauto{q_i}{r_i}$ for every $i \leq k$. This completes the proof of Proposition~\ref{prop:covesynch}.

\subsection{\ptime-completeness}
We prove that \grp-separation is \ptime-complete. We already proved that it is in \ptime. We show that it is \ptime-hard, even  when one of the two inputs is the singleton $\{\veps\}$.
We reduce the Monotone Circuit Value problem, a variant of the Circuit Value Problem in which all gates are either a disjunction ($\vee$) or a conjunction ($\wedge$). It is known to be \ptime-complete~\cite{Goldschlager1977TheMA}. Let us describe it.

\noindent A \emph{Boolean circuit} is a finite directed acyclic graph such that:
\begin{itemize}
  \item There are \emph{input vertices} with no incoming edge and labeled by truth values ($0$ for \emph{false}, $1$ for \emph{true}).
  \item The other vertices have exactly two incoming edges. They are called \emph{gates} and are labeled by a logical connective: ``$\vee$'' or ``$\wedge$''. They have arbitrarily many outgoing edges.
  \item There is a single gate with no outgoing edge. It is called the \emph{output vertex}.
\end{itemize}
We present an Example of a Boolean circuit in Figure~\ref{fig:circuit} below.

\tikzstyle{mstate}=[circle,minimum size=5mm,inner sep = 1pt,draw=black,thick,align=center]
\tikzstyle{medge}=[draw=black,thick,->]
\tikzstyle{mlabel}=[fill=white,inner sep=0pt,very thick,circle]

\begin{figure}[!htb]
  \centering
  \scalebox{0.85}{
    \begin{tikzpicture}[scale=.6]
      \node[mstate] (w) at (0.0,4.5) {$0$};
      \node[mstate] (x) at (0.0,3.0) {$1$};
      \node[mstate] (y) at (0.0,1.5) {$0$};
      \node[mstate] (z) at (0.0,0.0) {$1$};

      \node[mstate] (c11) at (2.0,3.75) {$\lor$};
      \node[mstate] (c12) at (2.0,2.25) {$\lor$};
      \node[mstate] (c13) at (2.0,0.75) {$\land$};

      \node[mstate] (c21) at (4.0,4.5) {$\land$};
      \node[mstate] (c22) at (4.0,1.5) {$\land$};
      \node[mstate] (c23) at (4.0,0.0) {$\lor$};

      \node[mstate] (c31) at (6.0,3.0) {$\vee$};
      \node[mstate] (c32) at (6.0,0.75) {$\vee$};

      \node[mstate] (c4) at (8.0,1.9) {$\land$};

      \draw[medge] (w) to [out=-10,in=165] node {} (c11);
      \draw[medge] (x) to [out=10,in=195] node {} (c11);
      \draw[medge] (x) to [out=-10,in=165] node {} (c12);
      \draw[medge] (y) to [out=10,in=195] node {} (c12);
      \draw[medge] (y) to [out=-10,in=165] node {} (c13);
      \draw[medge] (z) to [out=10,in=195] node {} (c13);

      \draw[medge] (w) to [out=0,in=180] node {} (c21);
      \draw[medge] (c11) to [out=0,in=205] node {} (c21);
      \draw[medge] (c12) to [out=0,in=165] node {} (c22);
      \draw[medge] (c13) to [out=10,in=195] node {} (c22);
      \draw[medge] (c13) to [out=-10,in=155] node {} (c23);
      \draw[medge] (z) to [out=0,in=180] node {} (c23);

      \draw[medge] (c21) to [out=0,in=165] node {} (c31);
      \draw[medge] (c22) to [out=10,in=195] node {} (c31);
      \draw[medge] (c22) to [out=-10,in=165] node {} (c32);
      \draw[medge] (c23) to [out=0,in=195] node {} (c32);

      \draw[medge] (c31) to [out=0,in=165] node {} (c4);
      \draw[medge] (c32) to [out=0,in=195] node {} (c4);
    \end{tikzpicture}}
  \caption{An example of a Boolean circuit, which evaluates to $0$.}
  \label{fig:circuit}
\end{figure}
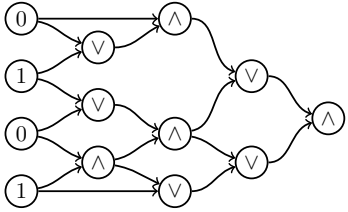

A Boolean circuit computes a truth value for each gate. The decision problem takes as input a Boolean circuit $C$ and asks if the value computed by the output vertex is true. We present a logarithmic space reduction from this problem to \emph{non}-separability by \grp. Given as input a Boolean circuit $C$, we construct an \nfa $\As_C$ such that $C$ evaluates to true if and only if $\{\veps\}$ is not \grp-separable from $L(\As_C)$. This implies that \grp-separation is \ptime-hard, as desired. We only present the construction of $\As_C$. That it can be implemented in logarithmic space is straightforward and left to the reader.

We fix $C$ and describe the \nfa $\As_C = (Q,I,F,\delta)$. We let $n$ be the number of vertices in $C$ and $\{v_1,\dots,v_n\}$ be the set of all these vertices, with $v_n$ as the output vertex. The \nfa \As uses an alphabet $A = \{a_1,\dots,a_n\}$ of size $n$. For each $i \leq n$, the set of states $Q$ contains three states $q_i,r_i$ and $s_i$ associated to the vertex $v_i$ (note that $s_i$ is only useful when $v_i$ is a gate labeled by ``$\wedge$''). Moreover, we also associate several transitions in $\delta$ connecting these three states to those associated to other vertices. There are several cases depending on $v_i$.

First, assume that $v_i$ is an input vertex. If $v_i$ is labeled by ``$0$'' (false), we add the following transition to $\As_C$:
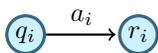
\begin{center}
  \scalebox{1.0}
  {
    \begin{tikzpicture}
      \node[lstate] (s) at (0.0,0.0) {$q_i$};
      \node[lstate] (r) at (1.5,0.0) {$r_i$};

      \draw[trans] (s) edge node[above] {$a_i$} (r);
    \end{tikzpicture}
  }
\end{center}
If $v_i$ is labeled by ``$1$'', we add the following transitions:
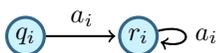
\begin{center}
  \scalebox{1.0}
  {
    \begin{tikzpicture}
      \node[lstate] (s) at (0.0,0.0) {$q_i$};
      \node[lstate] (r) at (1.5,0.0) {$r_i$};

      \draw[trans] (s) edge node[above] {$a_i$} (r);
      \draw[trans] (r) edge [loop right] node[right] {$a_i$} (r);
    \end{tikzpicture}
  }
\end{center}
Assume now that $v_i$ is a gate. Let $j,k \leq n$ be the two indices such that $C$ contains edges from $v_j$ to $v_i$ and from $v_k$ to $v_i$.  If $v_i$ is labeled by ``$\vee$'', we add the following transitions to $\As_C$:
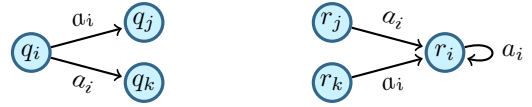
\begin{center}
  \scalebox{1.0}
  {
    \begin{tikzpicture}
      \node[lstate] (s) at (0.0,0.0) {$q_i$};
      \node[lstate] (sj) at (1.5,0.4) {$q_j$};
      \node[lstate] (sk) at (1.5,-0.4) {$q_k$};

      \node[lstate] (rj) at (4.0,0.4) {$r_j$};
      \node[lstate] (rk) at (4.0,-0.4) {$r_k$};
      \node[lstate] (r) at (5.5,0.0) {$r_i$};

      \draw[trans] (s) edge node[sloped,above] {$a_i$} (sj);
      \draw[trans] (s) edge node[sloped,below] {$a_i$} (sk);

      \draw[trans] (rj) edge node[sloped,above] {$a_i$} (r);
      \draw[trans] (rk) edge node[sloped,below] {$a_i$} (r);

      \draw[trans] (r) edge [loop right] node[right] {$a_i$} (r);
    \end{tikzpicture}
  }
\end{center}
If $v_i$ is labeled by ``$\wedge$'', we add the following transitions:
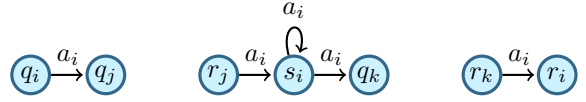
\begin{center}
  \scalebox{1.0}
  {
    \begin{tikzpicture}
      \node[lstate] (s) at (0.0,0.0) {$q_i$};
      \node[lstate] (sj) at (1,0.0) {$q_j$};

      \node[lstate] (rj) at (2.5,0.0) {$r_j$};
      \node[lstate] (q) at (3.5,0.0) {$s_i$};
      \node[lstate] (sk) at (4.5,0.0) {$q_k$};

      \node[lstate] (rk) at (6,0.0) {$r_k$};
      \node[lstate] (r) at (7,0.0) {$r_i$};

      \draw[trans] (s) edge node[above] {$a_i$} (sj);
      \draw[trans] (rj) edge node[above] {$a_i$} (q);

      \draw[trans] (q) edge node[above] {$a_i$} (sk);
      \draw[trans] (rk) edge node[above] {$a_i$} (r);

      \draw[trans] (q) edge [loop above] node[above] {$a_i$} (q);
    \end{tikzpicture}
  }
\end{center}

We let $\As_C = (Q,\{q_n\},\{r_n\},\delta)$. One may verify that the output vertex $v_n$ of $C$ evaluates to \emph{true} if and only if $\{\veps\}$ is not \grp-separable from $L(\As_C)$. Note that the proof argument does not look at \grp-separation directly: we use Theorem~\ref{thm:grpsep} instead. Indeed, it implies that $\{\veps\}$ is not \grp-separable from $L(\As_C)$ if and only if $\veps \in L(\econs{\As_C})$. It is straightforward to verify that the latter property holds if and only if the output vertex of $C$ evaluates to \emph{true}. One use induction to show that each gate $i$ evaluates to \emph{true} if and only if $\veps \in \aelauto{q_i}{r_i}$. This completes the presentation of our reduction.

\subsection{Connection with Ash's historical result}
\label{sec:history}
We compare Theorem~\ref{thm:grpsep} with the historical \grp-covering algorithm that can be deduced from Ash's results. We prove that the former is essentially a reformulation of the latter.

\newcommand{\opti}[2]{\ensuremath{\Is_{#1}[#2]}\xspace}
\newcommand{\copti}[1]{\opti{\Cs}{#1}}
\newcommand{\gropti}[1]{\opti{\grp}{#1}}

\smallskip
\noindent
{\bf Preliminaries.} Let \Cs be a Boolean algebra and $\alpha: A^* \to M$ be a morphism into a finite monoid. We define $\copti{\alpha} \subseteq 2^M$ as the set of all subsets $S \subseteq M$ such that $\{\alpha\inv(s) \mid s \in S\}$ is \emph{not} \Cs-coverable. It carries enough information to decide \Cs-covering for every input set consisting only of languages recognized by $\alpha$. More precisely, for $F_1,\dots,F_n \subseteq M$, one may verify that $\{\alpha\inv(F_i) \mid i \leq n\}$ is not \Cs-coverable if and only if there is $S \in \copti{\alpha}$ such that $S \cap F_i \neq \emptyset$ for all $i \leq n$.

Thus, a procedure computing $\copti{\alpha} \subseteq 2^M$ from an input morphism $\alpha: A^* \to M$ yields an algorithm for \Cs-covering.  Given a finite set of languages \Hb, one first computes a single morphism $\alpha: A^* \to M$ recognizing all $H \in \Hb$ (this is straightforward). Then, one computes $\copti{\alpha} \subseteq 2^M$. It carries enough information to decide whether \Hb is \Cs-coverable.

\smallskip
\noindent
{\bf Historical algorithm.} Ash's results~\cite{Ash91} yield a characterization of \gropti{\alpha}. We present this characterization (we use a formulation taken from~\cite{henckell:hal-00019815}) and prove that Theorem~\ref{thm:grpsep} is a natural reformulation on automata and a simple corollary.

We need weak inverses (they are the counterpart of automata construction $\As \mapsto \acauto$ of Section~\ref{sec:prelims}). Let $\alpha: A^*\to M$ be a morphism into a finite monoid. For $s \in M$, a \emph{weak inverse} of $s$ in an element $t \in M$ such that $tst = t$. We use this definition to associate a second morphism $\gamma_\alpha: \tilas \to 2^M$ over the extended alphabet \tila. For $a \in A$, we let,
\[
  \begin{array}{lll}
    \gamma_\alpha(a) & = &\{\alpha(a)\} \in 2^M, \\
    \gamma_\alpha(a\inv) &= & \{s \in \alpha(A^*) \mid \text{$s$ is a weak inverse of $\alpha(a)$}\}.
  \end{array}
\]
We now present the characterization. Recall that we write $L_\veps = \{w \in \tilas \mid w \xrightarrow{*} \veps\}$ (see Section~\ref{sec:grp}). We extend this notation to all words $u \in \tilas$: we let $L_u =  \{w \in \tilas \mid w \xrightarrow{*} u\}$.

\begin{theorem}[\cite{Ash91,henckell:hal-00019815}] \label{thm:ash}
  Let $\alpha: A^* \to M$ be a morphism into a finite monoid. Then, $\gropti{\alpha} \subseteq 2^M$ consists of all sets $\bigcup_{w \in L_u} \gamma_\alpha(w) \subseteq M$ for $u \in \tilas$.
\end{theorem}

\begin{remark}
  It is simple to verify that this yields an algorithm for computing $\gropti{\alpha}$ from $\alpha$. Roughly, one first needs to verify that the set $S_\veps=\bigcup_{w \in L_\veps} \gamma_\alpha(w)$ can be computed using a least fixpoint procedure (this is the counterpart of the construction $\As \mapsto \ecauto$ in Theorem~\ref{thm:grpsep}). Then, Theorem~\ref{thm:ash} implies that \gropti{\alpha} is the least subset of $2^M$ closed under multiplication and containing $S_\veps$ and all sets $\gamma_\alpha(b)$ for $b \in \tila$. It can be computed using again a least fixpoint procedure.
\end{remark}

Let us explain why Theorem~\ref{thm:ash} implies Theorem~\ref{thm:grpsep}. Let $k \geq 1$ and $\As_j = (Q_j,I_j,F_j,\delta_j)$ a \nfa for $1 \leq j \leq k$. Theorem~\ref{thm:grpsep} states that \mbox{$\{L(\As_j) \mid j \leq k\}$} is \grp-coverable if and only if  $\bigcap_{j \leq k} L(\econs{\As_j}) = \emptyset$. We use Theorem~\ref{thm:ash} to prove the right to left implication (the converse is simple as seen in Section~\ref{sec:grp}). Actually, we prove the contrapositive. Assume that $\{L(\As_j) \mid j \leq k\}$ is not \grp-coverable. We show that $\bigcap_{i \leq k} L(\econs{\As_i})\neq\emptyset$.

\newcommand{\popair}{\ensuremath{2^{\smash{Q^2}}}\xspace}

First, we build a morphism recognizing all languages $L(\As_j)$ (we use the standard transition morphism construction). Let $Q = \bigcup_{j \leq k} Q_j$ (we assume that the sets $Q_j$ are pairwise disjoint) and $\delta = \bigcup_{j \leq k} \delta_j$. Let $\As = (Q,\emptyset,\emptyset,\delta)$. Let $M = \popair$ be the monoid equipped with the standard multiplication: for $P,P' \in \popair$, we let $PP'$ as the set of all pairs $(q,s) \in Q^2$ such that $(q,r) \in P$ and $(r,s) \in P'$ for some $r \in Q$ (the set $\{(q,q) \mid q \in Q\}$ is the identity element). It is standard that the map $\alpha: A^* \to M$ defined by $\alpha(w) = \{(q,r) \mid w \in \alauto{q}{r}\}$ is a morphism. One may verify that $L(\As_j) = \alpha\inv(\{P \in \popair \mid P \cap (I_j \times F_j) \neq \emptyset\})$ for every $j \leq k$. The proof is based on the following simple lemma. It connects the \nfa \ecauto to weak inverses.

\begin{lemma} \label{lem:weakinv}
  Let $a \in A$ and let $P \in \alpha(A^*)$ be a weak inverse of $\alpha(a)$. For all $(q,r) \in P$, we have $a\inv \in \aelauto{q}{r}$.
\end{lemma}

\begin{proof}
  Let $P_a = \alpha(a)$ and $u \in \alpha\inv(P)$. By hypothesis, $PP_aP = P$. Thus $P_aP$ is idempotent.  Since $\alpha(au)= P_aP$, we get $\alpha((au)^n) = P_aP$ and $\alpha(u(au)^n) =PP_aP = P$ for all $n \geq 1$. As $(q,r) \in P$, this yields $u(au)^n \in\alauto{q}{r}$ for all $n\geq 1$. Thus, a pumping argument yields $s \in Q$ and $h,i,j \geq 1$ such that $u(au)^h \in \alauto{q}{s}$, $(au)^i \in \alauto{s}{s}$ and $(au)^j \in \alauto{s}{r}$. Since $\alpha((au)^n) = P_aP = \alpha(au)$ for all $n \geq 1$, we get $uau \in \alauto{q}{s}$, $au \in \alauto{s}{s}$ and $au \in \alauto{s}{r}$. The second property yields $t \in Q$ such that $a \in \alauto{s}{t}$ and $u \in \alauto{t}{s}$. Clearly, $s$ and $t$ are strongly connected. Hence, we get $u\inv \in \aclauto{s}{t}$ and $a\inv \in \aclauto{t}{s}$ by definition of \acauto. Altogether, we obtain that $w_1 = uauu\inv a\inv u\inv \in \aclauto{q}{t}$ and $w_2  = u\inv a\inv au \in \aclauto{s}{r}$. Since $w_i \xrightarrow{*} \veps$ for $i \in \{1,2\}$, we get the \veps-transitions $(q,\veps,t),(s,\veps,r) \in \ecdel$. Together with $a\inv \in \aclauto{t}{s}$, this yields $a\inv \in \aelauto{q}{r}$.
\end{proof}

We prove that $\bigcap_{j \leq k} L(\econs{\As_j})\neq\emptyset$. Since $\{L(\As_j) \mid j \leq k\}$ is not \grp-coverable, we get $S\! \in\! \gropti{\alpha}$ such that for all $j \leq k$, there is $P_j \in S$ such that $P_j\cap (I_j \times F_j) \neq \emptyset$. Theorem~\ref{thm:ash} yields $u \in \tilas$ such that $S = \bigcup_{w \in L_u} \gamma_\alpha(w)$. We use Lemma~\ref{lem:weakinv} to show that $u \in L(\econs{\As_j})$ for each $j \leq k$, completing the proof.

As $P_j \in S$, we get $w \in L_u$ such that $P_j \in \gamma_\alpha(w)$. We show that $w \in L(\econs{\As_j})$. As $w \xrightarrow{*} u$ (by definition of $L_u$), this yields $u \in L(\econs{\As_j})$ by Fact~\ref{fct:jump}. Let \mbox{$(q,r) \in P_j\cap (I_j \times F_j)$} and $b_1,\dots,b_n \in \tila$ such that $w = b_1 \cdots b_n$. As $(q,r) \in P_j$ and $P_j \in \gamma_\alpha(w) \subseteq M$, we get $T_1,\dots,T_n \in M = \popair$ and $q_0,\dots,q_n \in Q$ such that $q_0 = q$, $q_n = r$, $(q_{i-1},q_i) \in T_i$ and $T_i \in \gamma_\alpha(b_i)$ for $i \leq n$. We show that $b_i \in \aelauto{q_{i-1}}{q_i}$ for all $i \leq n$. This yields $w = b_1 \cdots b_n \in \aelauto{q}{r}$, \emph{i.e.} $w \in L(\econs{\As_j})$ as desired since $(q,r) \in I_j \times F_j$. Let $i \leq n$. If $b \in A$, then $\gamma_\alpha(b_i) = \{\alpha(b_i)\}$. Hence, $T_i = \alpha(b_i)$ and since $(q_{i-1},q_i) \in T_i$, we get $b_i \in \alauto{q_{i-1}}{q_i}$ by definition of $\alpha$. Otherwise, if $b_i = a\inv \in A\inv$, the fact that $T_i \in \gamma_\alpha(b_i)$ implies that $T_i \in \alpha(A^*)$ is a weak inverse of $\alpha(a)$. Hence, since $(q_{i-1},q_i) \in T_i$, Lemma~\ref{lem:weakinv} yields $b_i = a\inv \in \aelauto{q_{i-1}}{q_i}$, completing the proof.

\section{Covering for alphabet modulo testable languages}
\label{sec:abe}
We consider the alphabet modulo testable languages. We first prove formally that these are the languages that can be recognized by a commutative group (this will be useful later). We then prove that \abg-covering is decidable in Theorem~\ref{thm:abesep} below. Finally, we prove that \abg-separation and \abg-covering are \conptime-complete.

\subsection{Algebraic characterization of \abg}

For every number $d \geq 1$, we associate an equivalence $\sim_d$ over $A^*$ and use it to characterize the languages in \abg. Let $d \geq 1$ and $w,w' \in A^*$, we write $w \sim_d w'$ if and only if $|w|_a \equiv |w'|_a \bmod d$ for every $a \in A$. It is immediate from the definition that $\sim_d$ is an equivalent of finite index.

\begin{lemma} \label{lem:amteq}
  Let $L \subseteq A^*$. We have $L \in \abg$ if and only if there exists $d \geq 1$ such that $L$ is a union of $\sim_d$-classes.
\end{lemma}

\begin{proof}
  Assume first that $L \in \abg$: $L$ is built from finitely many languages $L^a_{q,r}$ (for $a \in A$ and $q,r \in \nat$ such that $r < q$) using only unions and intersections. Let $d$ be the least common multiplier of all numbers $q \geq 1$ used in these languages. We show that $L$ is a union of $\sim_d$-classes. Let $w,w' \in A^*$ such that $w \sim_d w'$. We prove that $w \in L \Leftrightarrow w' \in L$. Clearly, it suffices to show that each language $L^a_{q,r}$ used to define $L$ satisfies $w \in L^a_{q,r} \Leftrightarrow w' \in L^a_{q,r}$. Since $d$ is a multiple of $q$, the hypothesis that $w \sim_d w'$ yields $|w|_a \equiv |w'|_a \bmod q$. Hence, since $L^a_{q,r} = \{w \in A^* \mid |w|_a \equiv r \bmod q\}$ by definition, we have $w \in L^a_{q,r} \Leftrightarrow w' \in L^a_{q,r}$ as desired.

  Assume now that $L$ is a union of $\sim_d$-classes for $d \geq 1$. We show that $L \in \abg$. Since $\sim_d$ has finite index, it suffices to show that all $\sim_d$-classes belongs to \abg. Let $w \in A^*$ and consider its $\sim_d$-class. For every $a \in A$, let $r_a < d$ by the remainder of the Euclidean division of $|w|_a$ by $d$. By definition, for every $w ' \in A^*$, we have $w' \sim_d w$ if and only if $|w'|_a \equiv r_a \bmod d$ for every $a \in A$. It follows that the $\sim_d$-class of $w$ is $\bigcap_{a \in A} L^a_{d,r_a}$ which belongs to \abg by definition.
\end{proof}

We now prove the algebraic characterization of \abg.

\begin{restatable}{lemma}{amtcar} \label{lem:amtcar}
  The class \abg consists of all languages that are recognized by a morphism into a finite commutative group.
\end{restatable}

\begin{proof}
  First consider $L \in \abg$. We show that $L$ is recognized by a morphism into a finite commutative group. By definition of \abg, it suffices to prove that this property is true for all basic languages $L^a_{q,r}$ and that it is preserved by union and intersection. by definition $L^a_{q,r} = \{w \in A^* \mid |w|_a \equiv r \bmod q\}$ for $a \in A$ and $q,r \in \nat$ such that $r < q$. It is recognized by the morphism $\alpha: A^* \to {\integ}/q{\integ}$ (where ${\integ}/q{\integ} = \{0,\dots,q-1\}$ is the standard cyclic group) defined by $\alpha(a) = 1$ and $\alpha(b) = 0$ for $b \in A \setminus \{a\}$. We have $L^a_{q,r} = \alpha\inv(r)$. Finally, if $L_1,L_2 \subseteq A^*$ are such that $L_i$ is recognized by a morphism $\alpha_i: A^* \to G_i$ into a finite commutative group for $i \in \{1,2\}$, then $L_1 \cup L_2$ and $L_1 \cap L_2$ are recognized by the natural morphism $\alpha: A^* \to G_1 \times G_2$ (where $G_1 \times G_2$ is equipped with the componentwise multiplication).

  Assume now that $L$ is recognized by a morphism $\alpha\!: \!A^*\! \to\! G$ into a finite commutative group $G$. We show that $L\in\abg$. Since $G$ is a finite group, it is standard that there exists a number $d \geq 1$ such that $g^d = 1_G$ for every $g \in G$. We show that for every $u,v \in A^*$, it $u \sim_d v$, then $\alpha(u) = \alpha(v)$. It will follows that every language recognized by $\alpha$ is a union of $\sim_d$-classes and therefore belongs to \abg by Lemma~\ref{lem:amteq}. Recall that $A = \{a_1,\dots,a_n\}$. As $G$ is commutative, reorganizing the letters in $u,v$ does not change their image under $\alpha$. Thus,
  \[
    \alpha(u) = \alpha(a_1^{|u|_{a_1}} \cdots a_n^{|u|_{a_n}}) \text{ and } \alpha(v) = \alpha(a_1^{|v|_{a_1}} \cdots a_n^{|v|_{a_n}}).
  \]
  If $u \sim_d v$, then $|u|_{a_i} \equiv |v|_{a_i} \bmod d$ for every $i \leq n$. We get $r_i < d$ and $h_i,k_i \in \nat$ such that  $|u|_{a_i} = r_i + h_i \times d$ and $|v|_{a_i} = r_i + k_i \times d$. Therefore, since $g^d = 1_G$ for all $g \in G$, we obtain that  $\alpha(a_i^{\smash{|u|_{a_i}}}) = \alpha(a_i^{\smash{|v|_{a_i}}}) = \alpha(a_i^{\smash{r_i}})$. Altogether, we get $\alpha(u) = \alpha(v) = \alpha(a_1^{r_1} \cdots a_n^{r_n})$, concluding the proof.
\end{proof}

\subsection{Covering for \abg}
We prove that covering is decidable for \abg as well.  Let us point out that this can be obtained from an algebraic theorem of Delgado~\cite{abelianp}. Yet, this approach is indirect: Delgado's results are purely algebraic and do not mention separation. Formulating them would require a lot of groundwork. We use a direct approach based on standard arithmetical and automata theoretic arguments. As for \grp, we present a theorem characterizing the finite sets of regular languages which are \abg-coverable. We reuse the construction $\As \mapsto \acauto$ of Section~\ref{sec:prelims}. We start with terminology that we need to formulate the result.

Let $n = |A|$. Consider an arbitrary linear order $A$ and let $A = \{a_1,\dots, a_n\}$. We define a map $\zeta: \tilas \to \integ^n$ (where $\integ$ is the set of integers). Given, $w \in \tilas$, we define,
\[
  \zeta(w) = (|w|_{\smash{a_1}} - |w|_{\smash{a_1\inv}},\dots,|w|_{\smash{a_n}} - |w|_{\smash{a_n\inv}}) \in \integ^n.
\]
For a language $L \subseteq \tilas$ over \tila, we shall consider the direct image $\zeta(L) = \{\zeta(w) \mid w \in L\} \subseteq \integ^n$. We may now present the characterization theorem.

\begin{restatable}{theorem}{abesep} \label{thm:abesep}
  Let $k \geq 1$ and $k$ \nfas $\As_1,\dots,\As_k$. The following conditions are equivalent:
  \begin{enumerate}
    \item The set $\{L(\As_1),\dots,L(\As_k)\}$ is \abg-coverable.
    \item We have $\bigcap_{i \leq k} \zeta(L(\acons{\As_i})) = \emptyset$.
  \end{enumerate}
\end{restatable}

We first explain why Theorem~\ref{thm:abesep} implies the decidability of \abg-covering. This follows from standard results and the decidability of Presburger arithmetic. Let us present a sketch.

The definition of the map $\zeta: \tilas \to \integ^n$ is a variation on a standard notion. Given a word $w \in \tilas$, its \emph{Parikh image} (also called commutative image) is defined as the following vector,
\[
  \pi(w) = (|w|_{\smash{a_1}},\dots,|w|_{\smash{a_n}},|w|_{\smash{a_1\inv}},\dots,|w|_{\smash{a_n\inv}}) \in \nat^{2n}.
\]
Clearly, $\pi(w)$ determines $\zeta(w)$ and for every $L \subseteq \tilas$, $\pi(L) \subseteq \nat^{2n}$ determines $\zeta(L) \subseteq \integ^n$. Consider $k$ \nfas $\As_1,\dots,\As_k$. We know that $\acons{\As_i}$ can be computed from $\As_i$ in polynomial time for every $i \leq k$. Moreover, it is known~\cite{prescomput} that an \emph{existential} Presburger formula $\varphi_i$ describing the set $\pi(L(\acons{\As_i})) \subseteq \nat^{2n}$ can be computed from $\tilde{\As_i}$ in polynomial time. It is then straightforward to combine the formulas $\varphi_i$ into a single \emph{existential} Presburger sentence which is equivalent to $\bigcap_{i \leq k} \zeta(L(\acons{\As_i})) \neq \emptyset$. Finally, it is known~\cite{presNP1} that the existential fragment of Presburger arithmetic can be decided in \nptime. Hence, deciding whether $\bigcap_{i \leq k} \zeta(L(\acons{\As_i})) \neq \emptyset$ can be achieved in \nptime. It then follows from Theorem~\ref{thm:abesep} that \abg-covering (and therefore \abg-separation as well) can be decided in \conptime. It turns out that this complexity upper bound is optimal: \abg-covering and \abg-separation are both \conptime-complete (we present a simple proof for the lower bound using a reduction from \tsat).

\begin{proof}[Proof of Theorem~\ref{thm:abesep}]
  We fix a number $k \geq 1$ and for every $j \leq k$, we consider an \nfa $\As_j = (Q_j,I_j,F_j,\delta_j)$. The two implications in the theorem are handled independently.

  \medskip
  \noindent
  {\bf Implication $1) \Rightarrow 2)$.} We prove the contrapositive. Consider $\overbar{v} \in \bigcap_{i \leq k} \zeta(L(\acons{\As_i}))$. We show that $\{L(\As_1),\dots,L(\As_k)\}$ is \emph{not} \abg-coverable. Thus, we fix an \abg-cover \Kb of $A^*$ and show that there exists $K \in \Kb$ such that $K \cap L(\As_j) \neq \emptyset$ for every $j \leq k$. Let $\{K_1,\dots,K_\ell\} = \Kb$. For each $i \leq \ell$, since $K_i \in \abg$,  Lemma~\ref{lem:amtcar} yields a morphism $\alpha_i: A^* \to G_i$ into a finite \emph{commutative} group recognizing $K_i$. Clearly, \mbox{$G = G_1 \times \cdots \times G_\ell$} is a commutative group for the componentwise multiplication and each $K \in \Kb$ is recognized by the morphism $\alpha: A^* \to G$ defined by $\alpha(w) = (\alpha_1(w),\dots,\alpha_n(w))$.

  Since $\overbar{v} \in \bigcap_{i \leq k} \zeta(L(\acons{\As_i}))$, we get $w_i \in L(\acons{\As_i})$ such that $\zeta(w_i) = \overbar{v}$ for all $i \leq \ell$. Also, Lemma~\ref{lem:trans} yields $u_i \in L(\As_i)$ such that $\alpha(u_i) = \alpha(w_i)$. As $G$ is \emph{commutative}, the image under $\alpha$ of a word $w \in \tilas$ depends only on $\zeta(w)$. Hence, $\alpha(w_1) = \cdots = \alpha(w_k)$. We get $\alpha(u_1) = \cdots = \alpha(u_k)$. As \Kb is a cover of $A^*$, we get $K \in \Kb$ such that $u_1 \in K$. Since $K$ is recognized by $\alpha$ and $\alpha(u_1)=\cdots = \alpha(u_k)$, this yields $u_1,\dots,u_k \in K$. Thus, $K \cap L(\As_i) \neq \emptyset$ for all $i \leq \ell$ as desired.

  \medskip
  \noindent
  {\bf Implication $2) \Rightarrow 1)$.} We use standard arithmetical tools.  Consider the componentwise addition on $\integ^n$. We abuse notation and write ``$0$'' for the identity element (\emph{i.e.}, the vector whose entries are all equal to zero). For a single vector $\overbar{v} \in \integ^{n}$ and a \emph{finite} set of vectors $V = \{\overbar{v_1},\dots,\overbar{v_\ell}\} \subseteq \integ^{n}$, we write,
  \[
    \Ls(\overbar{v},V) =  \{\overbar{v} + k_1\overbar{v_1} + \cdots + k_\ell\overbar{v_\ell} \mid k_1,\dots,k_\ell \in \integ\} \subseteq \integ^{n}.
  \]
  Following~\cite{DBLP:journals/dmtcs/ChoffrutF10}, we call these sets the \emph{\integ-linear subsets} of~$\integ^{n}$. Likewise, \emph{\integ-semilinear} subsets are finite unions of \integ-linear sets (including~$\emptyset$, which is the empty union). We need two results about these sets. The first one is a variation on Parikh's theorem (which implies that the Parikh images of regular languages are semilinear subsets of $\nat^n$). It is specific to the automata built with $\As\mapsto \acauto$.

  \begin{restatable}{lemma}{autozeta}\label{lem:autozeta}
    Let \As be an \nfa. Then, $\zeta(L(\acauto))$ is a \integ-semilinear subset of $\integ^n$.
  \end{restatable}

  \begin{proof}
    The proof is based on standard ideas which are typically used to prove the automata variant of Parikh's theorem. However, let us point out that we do require a specific property of the automaton \acauto at some point (the lemma is not true for an arbitrary \nfa over the extended alphabet \tila). For all $q \in Q$, we associate a finite set $V_q \subseteq \integ^n$. We define,
    \[
      V_q = \{\zeta(w) \mid w \in \aclauto{q}{q} \text{ and } |w| \leq |Q|\}.
    \]
    Observe that if $w \in \aclauto{q}{q}$ for some $w \in \tilas$, the states encountered on this run are strongly connected. Hence, in that case, we also have $w\inv \in \aclauto{q}{q}$ by definition of \acauto. Moreover, we have $\zeta(w\inv) = - \zeta(w)$ by definition of $\zeta$. Consequently, for every $\overbar{v} \in V_q$, the opposite vector also belongs to $V_q$: we have $-\overbar{v} \in V_q$. This property is where we need the hypothesis that are considering an automata built with the construction $\As \mapsto \acauto$ (it fails for an arbitrary \nfa). For every $P \subseteq Q$, we write $V_P = \bigcup_{q \in  P} V_q$.

    Finally, we associate a second finite set $X_P \subseteq \integ^n$ to every $P \subseteq Q$. Let $w \in \tilas$. We say that $w$ is a $P$-witness if there exist $q\in I$ and $r \in F$ such that there is a run from $q$ to $r$ labeled by $w$ such that $P$ is exactly the set of all states encountered in that run (in particular, we have $w \in L(\acauto)$). We define,
    \[
      X_P = \{\zeta(w) \mid \text{$w$ is a $P$-witness and $|w| \leq |Q|^2$}\}.
    \]
    We now prove the following,
    \[
      \zeta(L(\acauto)) = \bigcup_{P \subseteq Q} \bigcup_{\overbar{v} \in X_P} \Ls(\overbar{v},V_P).
    \]
    This equality concludes the proof: $\zeta(L(\acauto))$ is a \integ-semilinear subset of $\integ^n$, as desired. We start with the right to left inclusion. Let $P \subseteq Q$ and $\overbar{v} \in X_P$. We show that $\Ls(\overbar{v},V_P) \subseteq \zeta(L(\acauto))$.

    Let $\overbar{u} \in \Ls(\overbar{v},V_P)$. By definition, we have $\overbar{v_1},\dots,\overbar{v_\ell} \in V_P$ and $k_1,\dots,k_\ell \in \integ$ such that $\overbar{u} = \overbar{v} + k_1\overbar{v_1} + \cdots + k_\ell\overbar{v_\ell}$. Moreover, recall that by construction, for every $\overbar{v_i}$, the opposite vector $-\overbar{v_i}$ belongs to $V_P$ as well. Therefore, we may assume without loss of generality that $k_1,\dots,k_\ell \in \nat$: they are \emph{positive integers}.  By definition of $V_P$, we know that for every $i \leq \ell$, we have $\overbar{v_i} \in V_{q_i}$ for some $q_i \in P$. Hence, there exists $x_i \in \aclauto{q_i}{q_i}$ such that $\zeta(x_i) = \overbar{v_i}$. Let $y_i = (x_i)^{k_i}$ (this is well-defined since $k_i \in \nat$). Clearly, $\zeta(y_i) = k_i\overbar{v_i}$ and $y_i \in \aclauto{q_i}{q_i}$. Moreover, $\overbar{v} \in X_P$ which yields a $P$-witness $w \in \tilas$ such that $\zeta(w) = \overbar{v}$. Since $q_1,\dots,q_\ell \in P$ and $w$ is a $P$-witness, we have $q\in I$ and $r \in F$ such that there exists a run from $q$ to $r$ labeled by $w$ which encounters all states $q_1,\dots,q_\ell$. Therefore, we have a permutation $\sigma$ of $\{1,\dots,\ell\}$ and $w_0,\dots,w_\ell \in \tilas$ such that \mbox{$w = w_0 \cdots w_\ell$}, $w_0 \in \aclauto{q}{q_{\sigma(1)}}$, $w_i \in \aclauto{q_{\sigma(i)}}{q_{\sigma(i+1)}}$ for every $1 \leq i \leq n-1$ and $w_\ell \in \aclauto{q_{\sigma(\ell)}}{r}$. Consider the word $w' = w_0y_{\sigma(1)}w_1 \cdots y_{\sigma(\ell)}w_\ell$. It is clear from the definitions that $w' \in \aclauto{q}{r}$ which yields \mbox{$w' \in L(\acauto)$} and $\zeta(w') \in \zeta(L(\acauto))$. Moreover, it is immediate that \mbox{$\zeta(w') = \zeta(w) + \zeta(y_1) + \cdots + \zeta(y_\ell)$} which yields \mbox{$\zeta(w') = \overbar{v} + k_1\overbar{v_1} + \cdots + k_\ell\overbar{v_\ell} = \overbar{u}$}. We get $\overbar{u} \in \zeta(L(\acauto))$ as desired.

    We turn to the converse inclusion which is based on pumping arguments. Given a word $w \in L(\acauto)$, we need to prove that $\zeta(w) \in \bigcup_{P \subseteq Q} \bigcup_{\overbar{v} \in X_P} \Ls(\overbar{v},V_P)$. Since $w \in  L(\acauto)$, there exists $q \in I$ and $r \in F$ such that $w \in \aclauto{q}{r}$. We let $P \subseteq Q$ be the set of all states which are encountered in the corresponding run: $w$ is a $P$-witness. We use induction on the length of $w$ to show that there exists $\overbar{v} \in X_P$ such that $\zeta(w) \in\Ls(\overbar{v},V_P)$ (which concludes the argument). There are two cases. First assume that $|w| \leq |Q|^2$. This implies $\zeta(w) \in X_P$ by definition and we have $\zeta(w) \in \Ls(\zeta(w),V_P)$, concluding this case. Assume now that $|w| > |Q|^2$. One may verify with a pumping argument that there exist $x_1,x_2 \in A^*$ and $y \in A^+$ such that $w = x_1yx_2$, the word $w' = x_1x_2$ remains a $P$-witness, $|y| \leq |Q|$ and $y \in \aclauto{q}{q}$ for some $q \in P$.  Since $y \in A^+$ and $w = x_1yx_2$, we have \mbox{$|w'| < |w|$}. Thus, since $w'$ is a $P$-witness, induction yields $\overbar{v} \in X_P$ such that $\zeta(w') \in\Ls(\overbar{v},V_P)$. Moreover, since $|y| \leq |Q|$ and \mbox{$y \in \aclauto{q}{q}$} for some $q \in P$, we have $\zeta(y) \in V_P$ by definition. Thus, \mbox{$\zeta(w') + \zeta(y)\in \Ls(\overbar{v},V_P)$}. Finally, since  $w = x_1yx_2$ and $w' = x_1x_2$, it is clear that $\zeta(w) = \zeta(w') + \zeta(y)$. Altogether, we obtain $\zeta(w) \in \Ls(\overbar{v},V_P)$ which concludes the~proof.
  \end{proof}

  The second result is more general.

  \begin{restatable}{proposition}{fab} \label{prop:fab}
    Let $n \geq 1$ and $S \subseteq \integ^n$ be a \integ-semilinear set. Assume that for all $d \geq 1$, there exists a vector $\overbar{u} \in \integ^n$ such that $d\overbar{u} \in S$. Then, $0 \in S$.
  \end{restatable}

  Proposition~\ref{prop:fab} is a corollary of a standard theorem about  bases of subgroups of free Abelian groups (\emph{i.e.}, the groups $\integ^n$). We first introduce terminology that we need to state this theorem. Clearly, $\integ^n$ is a commutative group for addition (called ``free abelian group of rank $n$''). We consider the subgroups of $\integ^n$ (the subsets which are closed under addition and inverses). Additionally, we need the notion of \emph{basis}. Given a subgroup $G$ of $\integ^n$, a basis of $G$ is a finite set of vectors $\{\overbar{v_1},\dots,\overbar{v_m}\} \subseteq G$ which satisfies the two following conditions:

  \begin{enumerate}
    \item $G$ is generated by $\{\overbar{v_1},\dots,\overbar{v_m}\}$. That is, we have \mbox{$G =  \{k_1\overbar{v_1} + \cdots + k_m\overbar{v_m} \mid k_1,\dots,k_m \in \integ\}$}.
    \item For all $k_1,\dots,k_m \in \integ$ such that $k_1\overbar{v_1} + \cdots + k_m\overbar{v_m} = 0$, we have $k_1 = \cdots = k_m = 0$.
  \end{enumerate}

  We now state the following standard theorem (see for example \cite[Theorem~1.6]{hungerford}).

  \begin{theorem}\label{thm:abg:fab}
    Let $G$ be a nontrivial subgroup of $\integ^n$. There exist a basis $\{\overbar{x_1},\dots,\overbar{x_n}\}$ of~$\integ^n$, a number $m \leq n$ and $d_1,\dots,d_m \geq 1$ such that $d_i$ divides $d_{i+1}$ for every $i \leq m-1$ and $\{d_1\overbar{x_1},\dots,d_m\overbar{x_m}\}$ is a basis of $G$.
  \end{theorem}

  We are now ready to prove Proposition~\ref{prop:fab}.
  \begin{proof}[Proof of Proposition~\ref{prop:fab}]
    Observe first that we may assume without loss of generality that $S$ is a \integ-linear subset of $\integ^n$. Indeed, by definition $S$ is a \emph{finite} union of \integ-linear subsets. Hence, by hypothesis, for every $d \geq 1$, there exists $\overbar{u} \in \integ^n$ and a \integ-linear set $S'$ in this union such that $d\overbar{u}\in S'$. In particular, this is true when $d= h!$ for some $h \geq 1$. Hence, since the union is finite, it contains a fixed \integ-linear set $S'$ such that there exists infinitely many $d$ such that $d = h!$ for some $h \geq 1$ and $d\overbar{u}\in S'$. It then follows that for every $d \geq 1$, there exists $\overbar{u} \in \integ^n$ such that $d\overbar{u}\in S'$. Therefore, we may replace $S$ with~$S'$.

    We assume from now on that $S$ is \integ-linear: we have $\overbar{v} \in \integ^n$ and a finite set $V \subseteq \integ^n$ such that $S = \Ls(\overbar{v},V)$. If $V = \emptyset$ or $V = \{0\}$, we have $\Ls(\overbar{v},V) = \{\overbar{v}\}$. Thus, for every $d \geq 1$, there exists $\overbar{u} \in \integ^n$ such that $\overbar{v} = d\overbar{u}$. In particular, this holds for a number $d$ which is strictly larger than the absolute values of all entries in $\overbar{v}$. Clearly, this implies $\overbar{v} = 0$ and we get $0 \in \Ls(\overbar{v},V)$. We now assume that $V$ contains a non-zero vector.

    Let $G \subseteq \integ^n$ be the subgroup of $\integ^n$ generated by the set $V \subseteq \integ^n$. By hypothesis on $V$, $G$ is nontrivial. Therefore, Theorem~\ref{thm:abg:fab} yields a basis $\{\overbar{x_1},\dots,\overbar{x_n}\}$ of $\integ^n$, $m \leq n$ and $d_1,\dots,d_m \geq 1$ such that $\{d_1\overbar{x_1},\dots,d_m\overbar{x_m}\}$ is a basis of $G$.

    Since $\{\overbar{x_1},\dots,\overbar{x_n}\}$ is a basis of $\integ^n$, we have $h_1,\dots,h_n \in \integ$ such that $\overbar{v} = h_1\overbar{x_1} + \cdots + h_n\overbar{x_n}$. Let $d$ be the least common multiplier of $|h_1|+1,\dots,|h_n|+1,d_1,\dots,d_m \geq 1$. By hypothesis, there exists $\overbar{u} \in \integ^n$ such that $d\overbar{u} \in \Ls(\overbar{v},V)$. Thus, since $G$ is the subgroup generated by~$V$, there exists $\overbar{y} \in G$ such that $d\overbar{u} = \overbar{v} + \overbar{y}$. Since $\{\overbar{x_1},\dots,\overbar{x_n}\}$ is a basis of $\integ^n$, we have $k_1,\dots,k_n \in \integ$ such that $\overbar{u} = k_1\overbar{x_1} + \cdots + k_n\overbar{x_n}$. Moreover, since $\{d_1\overbar{x_1},\dots,d_m\overbar{x_m}\}$ is a basis of $G$, we have $\ell_1,\dots,\ell_m \in \integ$ such that $\overbar{y} = \ell_1 d_1\overbar{x_1} + \cdots + \ell_m d_m\overbar{x_m}$. Altogether, we obtain,
    \[
      \sum_{1 \leq i \leq n} h_i\overbar{x_i} + \sum_{1 \leq j \leq m}  \ell_j d_j\overbar{x_j}= \sum_{1 \leq i \leq n}  dk_i\overbar{x_i}.
    \]
    Since $\{\overbar{x_1},\dots,\overbar{x_n}\}$ is a basis, this implies that for all $i > m$, we have $h_i = dk_i$. By definition $d > |h_i|$ (it is a nonzero multiple of $|h_{i}|+1$). Thus,  $dk_i = h_i$ implies that $k_i = h_i = 0$. Since this holds for every $i > m$, we obtain,
    \[
      \sum_{1 \leq i \leq n} h_i\overbar{x_i} + \sum_{1 \leq j \leq m}  \ell_j d_j\overbar{x_j}= \sum_{1 \leq i \leq m}  dk_i\overbar{x_i}.
    \]
    This yields the following,
    \[
      \overbar{v} + (\ell_1 d_1 - dk_1)\overbar{x_1} + \cdots + (\ell_m d_m - dk_m) \overbar{x_m} = 0.
    \]
    By definition $d$ is a multiple of $d_i$ for every $i \leq m$. Therefore, there exists $\ell'_i \in \integ$ such that $\ell_i d_i - dk_i = \ell'_i d_i$. Thus, we obtain,
    \[
      \overbar{v} + \ell'_1 d_1\overbar{x_1} + \cdots + \ell'_m d_m \overbar{x_m} = 0.
    \]
    Since $\{d_1\overbar{x_1},\dots,d_m\overbar{x_m}\}$ is a basis of $G$ which is the subgroup generated by $V$, we obtain $0\in F(\overbar{v},V)$ as desired.
  \end{proof}

  We may now prove that $2) \Rightarrow 1)$ in Theorem~\ref{thm:abesep}. We consider the contrapositive. Assume that $\{L(\As_1),\cdots,L(\As_k)\}$ is \emph{not} \abg-coverable. We prove that $\bigcap_{j \leq k} \zeta(L(\acons{\As_j})) \neq \emptyset$. First, we use our hypothesis to prove the following lemma.

  \begin{lemma} \label{lem:comp}
    For every $d \geq 1$, there exist $\overbar{x},\overbar{y_1},\dots,\overbar{y_k} \in \integ^n$ such that $\overbar{x} + d\overbar{y_j} \in \zeta(L(\acons{\As_i}))$ for every $j \leq k$.
  \end{lemma}

  \begin{proof}
    Given $w,w' \in A^*$, we write $w \sim_d w'$ if and only if $|w|_a \equiv |w'|_a \bmod d$ for all $a \in A$. Clearly, $\sim_d$ is an equivalence of finite index on $A^*$. One may verify that each $\sim_d$-class belongs to \abg. Thus, the partition \Kb of $A^*$ into $\sim_d$-classes is an \abg-cover of $A^*$ and since $\{L(\As_1),\cdots,L(\As_k)\}$ is not \abg-coverable, there exists a $\sim_d$-class which intersects $L(\As_j)$ for all $j \leq k$. We get $w_1 \in L(\As_1),\dots,w_k \in L(\As_k)$ such that $w_1 \sim_d \cdots \sim_d w_n$. Let $\overbar{x} = \zeta(w_1) \in \nat^n$. Let $j \leq k$. The fact that $w_j \sim_d w_1$ yields $\overbar{y_j} \in \integ^n$ such that $\zeta(w_j) = \overbar{x} + d \overbar{y_j}$. Since $w_j \in L(\As_j) \subseteq L(\acons{\As_j})$, we have $\zeta(w_j) \in \zeta(L(\acons{\As_j}))$ which completes the proof.
  \end{proof}

  By Lemma~\ref{lem:autozeta}, $\zeta(L(\acons{\As_j})) \subseteq \integ^n$ is \integ-semilinear for each $j \leq k$. We build a \integ-semilinear subset of $\integ^{kn}$. We use vector concatenation: for $i_1,i_2 \geq 1$, $\overbar{x} \in \integ^{i_1}$ and $\overbar{y} \in \integ^{i_2}$, we write $\overbar{x} \cdot \overbar{y} \in \integ^{i_1+i_2}$ for the vector obtained by concatenating $\overbar{x}$ with~$\overbar{y}$. Let $S \subseteq \integ^{kn}$ be the set of all vectors $\overbar{u_1} \cdots \overbar{u_k} + \overbar{x}^k$ such that $\overbar{u_j} \in \zeta(L(\acons{\As_j}))$ for every $j \leq k$ and $\overbar{x} \in \integ^n$.

  Since the sets $\zeta(L(\acons{\As_j})) \subseteq \integ^n$ are \integ-semilinear, one may verify that $S \subseteq \integ^{kn}$ is \integ-semilinear as well. Lemma~\ref{lem:comp} implies that for every $d \geq 1$, there exist $\overbar{x},\overbar{y_1},\dots,\overbar{y_k} \in \integ^n$ such that $\overbar{x} + d\overbar{y_j} \in \zeta(L(\acons{\As_j}))$ for all $j \leq k$. By definition of $S$, this implies $d(\overbar{y_1} \cdots \overbar{y_k}) \in S$. Altogether, it follows that for all $d \geq 1$, there exists $\overbar{y} \in \integ^{kn}$ such that $d\overbar{y} \in S$. Since $S$ is \integ-semilinear, this yields $0 \in S$ by Proposition~\ref{prop:fab}. By definition of $S$, we get $\overbar{x} \in \integ^n$ such that $\overbar{x} \in \zeta(L(\acons{\As_j}))$ for all $j \leq k$. Thus, $\bigcap_{j \leq k} \zeta(L(\acons{\As_j})) \neq \emptyset$ which completes the proof.
\end{proof}

\subsection{Complexity lower bound}

We prove that \abg-covering and \abg-separation are \conptime-complete. As we explained above, the upper bound follows from Theorem~\ref{thm:abesep}. Here, we prove the lower bound: both problems are \conptime-hard. Actually since separation is a special case of covering, it suffices to show that \abg-separation is \conptime-hard.

\begin{remark}
  When considering complexity, it is important to distinguish the case when the alphabet is fixed from the one when it is a parameter of the problem. Here, we consider the latter case: we show that given an alphabet $A$ and two \nfas over $A$, deciding whether the recognized languages are \abg-separable is \conptime-hard. Actually, when the alphabet is fixed, one may show that the problem is in \ptime (roughly, this boils down to disjointedness of Parikh images for \nfas which is known to be in \ptime when the alphabet is fixed~\cite{parikhdecid}).
\end{remark}

We actually show that \emph{non} \abg-separability is \nptime-hard. More precisely, we present a logarithmic space reduction from \textup{3}-satisfiability (\tsat) to this problem. Given a \tsat formula $\varphi$, we explain how to construct two regular languages $L_1,L_2$ and show that they are not \abg-separable if and only if $\varphi$ is satisfiable. We only describe the construction: that \nfas for the regular languages $L_1$ and $L_2$ can be computed from $\varphi$ in logarithmic space is straightforward and left to the reader.

Let $C_1,\dots,C_k$ be the $3$-clauses such that $\varphi = \bigwedge_{i \leq k} C_i$ and let $x_1,\ldots,x_n$ be the propositional variables in $\varphi$. We construct two \emph{finite} languages $L_1$ and $L_2$ over the alphabet \mbox{$A = \{x_1, \ldots ,x_n,\overbar{x_1}, \ldots, \overbar{x_n}\}$}. Intuitively, we code assignments of truth values for the variables $\{x_1,\ldots,x_n\}$ by words in~$A^*$. Given $w \in A^*$, we say that $w$ \emph{is an encoding} if for all $i \leq n$, $w$ contains either the letter $x_i$ or the letter $\overbar{x_i}$, but not both. It is immediate that an assignment of truth values for the variables $\{x_1,\ldots,x_n\}$ can be uniquely defined from any such~encoding.

We let $H_i = \{x_i^p \mid 1 \leq p \leq k\} \cup \{\overbar{x_i}^p \mid 1 \leq p \leq k\}$ for all $i \leq n$. We may now define $L_1 \subseteq A^*$. We let,
\[
  L_1 = H_1H_2 \cdots H_n.
\]
Clearly $L_1$ is finite and all the words in $L_1$ are encodings. We turn to the definition of $L_2$. For every $j \leq k$, we associate a language $T_j$ to the $3$-clause $C_j$. Assume that $C_j = \ell_1 \vee \ell_2 \vee \ell_3$ where $\ell_1,\ell_2,\ell_3 \in \{x_1, \overbar{x_1},\ldots, x_n,\overbar{x_n}\}$ are literals. We define,
\[
  T_j = \{\ell_1,\ell_2,\ell_3\}.
\]
Finally, we define,
\[
  L_2 = T_1 \cdots T_k(\{\veps\} \cup H_1) \cdots (\{\veps\} \cup H_n).
\]
Clearly, $L_2$ is finite as well. Observe that the words in $L_2$ need not be encodings. On the other hand, all encodings within $L_2$ (if any) correspond to an assignment of truth values which satisfies $\{C_1,\dots,C_k\}$.

\smallskip

It remains to show that $L_1,L_2$ are not \abg-separable if and only if the $\varphi$ is satisfiable. We start with the right to left implication. Assume that there exists a truth assignment satisfying $\varphi$. By definition of $L_1$ and $L_2$, one may verify that there exists $w_1 \in L_1$ and $w_2 \in L_2$ which are both encodings of this assignment. Moreover, one may verify that we can choose $w_1$ and $w_2$ so that $|w_1|_a = |w_2|_a$ for every  $a \in A$. This  implies that $\alpha(w_1) = \alpha(w_2)$ for every morphism $\alpha: A^* \to G$ into an commutative group $G$. Hence, in view of Lemma~\ref{lem:amtcar}, every language $K \in \abg$ which contains $w_1$ must contain $w_2$ as well. Since $w_1 \in L_1$ and $w_2 \in L_2$, it follows that $L_1$ and $L_2$ are not \abg-separable.

Conversely, assume that $L_1$ and $L_2$ are not \abg-separable. By definition, $L_1$ and $L_2$ are finite. Thus, there exists $d \in \nat$ such that $|w| < d$ for every $w \in L_1 \cup L_2$. We consider the equivalence $\sim_d$ over $A^*$. By Lemma~\ref{lem:amteq}, every union of $\sim_d$-classes belongs to \abg. Hence, since $L_1$ and $L_2$ are not \abg-separable, there exists a $\sim_d$-class which intersects both $L_1$ and $L_2$. We obtain $w_1 \in L_1$ and $w_2 \in L_2$ such that $w_1 \sim_d w_2$: we have $|w_1|_a \equiv |w'_2|_a \bmod d$  for every $a \in A$. Moreover, since $|w_1| < d$ and $|w_2| < d$ by definition of $d$, this yields $|w_1|_a  = |w_2|_a $ for every $a \in A$. By definition of $L_1$, the word $w_1 \in L_1$ encodes an assignment of truth values. Moreover, since $|w_1|_a  = |w_2|_a$ for every $a \in A$, the word $w_2$ encodes the same assignment of truth values. Finally, since $w_2 \in L_2$, this assignment satisfies $\varphi$ which completes the proof.

\section{Covering for modulo languages}
\label{sec:mod}
In this section, we reduce \md-covering to \grp-covering and \abg-covering for unary alphabets. Then, we show that \md-covering is \conptime-complete, while \md-separation is \nl-complete.

\subsection{\md-covering algorithm}

Getting a ``naive'' direct algorithm for \md-covering is straightforward. Here, we prove that \md-covering reduces to both \grp-covering and \abg-covering. This approach provides much better complexity upper bounds than the naive one.

The reduction is based on a simple construction. It takes a language $L \subseteq A^*$ as input and builds a new one over a \emph{unary alphabet} (\emph{i.e.}, containing a unique letter). We let $U\! =\! \{\$\}$ and $\mu: A^* \to U^*$ be the morphism defined~by~$\mu(a) = \$$ for every $a \in A$. It is standard that if $L\!\subseteq\! A^*$ is recognized by an input \nfa \As, one may compute an \nfa recognizing $\mu(L)$ in logarithmic space (all transitions must be relabeled by ``$\$$'').

\begin{restatable}{theorem}{mdcov} \label{thm:mdcov}
  Let $k \geq 1$ and $L_1,\dots,L_k \subseteq A^*$. The following conditions are equivalent:
  \begin{enumerate}
    \item The set $\{L_1,\dots,L_k\}$ is \md-coverable.
    \item The set $\{\mu(L_1),\dots,\mu(L_k)\}$ is \abg-coverable.
    \item The set $\{\mu(L_1),\dots,\mu(L_k)\}$ is \grp-coverable.
  \end{enumerate}
\end{restatable}

\begin{proof}
  We prove that $1) \Rightarrow 2) \Rightarrow 3) \Rightarrow 1)$. Let us start with $1) \Rightarrow 2)$. Assume that $\{L_1,\dots,L_k\}$ is \md-coverable. We get a \md-cover \Kb of $A^*$ such that for each $K \in \Kb$, there is $i \leq k$ satisfying $K \cap L_i = \emptyset$. Let $\Hb = \{\mu(K) \mid K \in \Kb\}$. Since \Kb is a cover of $A^*$ and $\mu$ is surjective, \Hb must be a cover of $U^*$. One may verify that all $H \in \Hb$ belongs to \md since this is the case for all $K \in \Kb$. Hence, since $\md \subseteq \abg$, we obtain that \Hb is an \abg-cover of $U^*$. It remains to verify for each $H \in \Hb$, there exists $i \leq k$ such that $H \cap \mu(L_i) = \emptyset$. By definition, $H = \mu(K)$ for some $K \in \Kb$. By hypothesis on \Kb, we get $i \leq k$ such that $K \cap L_i = \emptyset$. We show that $H \cap \mu(L_i)=\emptyset$. By contradiction, assume that there exists $u \in H \cap \mu(L_i)$. As $H = \mu(K)$, we get $w \in K$  and $w' \in L_i$ such that $\mu(w) = \mu(w') = u$. By definition of~$\mu$, we have $|w| = |w'| = |u|$. Since $w \in K$ and $K \in \md$, this yields $w' \in K$. Thus, $w' \in K \cap L_i$, a contradiction.

  The implication $2) \Rightarrow 3)$ is trivial as $\abg \subseteq \grp$. It remains to prove $3) \Rightarrow 1)$. Assume that $\{\mu(L_1),\dots,\mu(L_k)\}$ is \grp-coverable. This yields a \grp-cover \Hb of $U^*$ such that for each $H \in \Hb$, there exists $i \leq k$ satisfying $H \cap \mu(L_i) = \emptyset$. We let $\Kb = \{\mu\inv(H) \mid H \in \Hb\}$. By definition of \Hb, one may verify that \Kb is a cover of $A^*$ and that for all \mbox{$K \in \Kb$}, there is $i \leq k$ such that $K \cap L_i = \emptyset$. It remains to show that \Kb is a \md-cover (which implies that $\{L_1,\dots,L_k\}$ is \md-coverable, as desired). Let $H \in \Hb$. We prove that $\mu\inv(H) \in \md$. By definition, we have to exhibit $q \geq 1$ such that for $w,w' \in A^*$, if $|w| \equiv |w'| \bmod q$, then $w\! \in\! \mu\inv(H) \Leftrightarrow w'\! \in\! \mu\inv(H)$ (\emph{i.e.}, $\mu(w) \in H \Leftrightarrow \mu(w') \in H$). Since $H \in \grp$, we get a morphism $\alpha: U^* \to G$ into a finite group $G$ recognizing $H$. It is standard that there is $q \geq 1$ such that $g^q = 1_G$ for all $g \in G$. We now fix $w,w' \in A^*$ such that $|w| \equiv |w'| \bmod q$. This yields $r \geq 0$ and $k,k' \geq 1$ such that $|w| = r + kq$ and $|w'| = r + k'q$. Hence, we have  $\mu(w) = \$^{r + qk}$ and $\mu(w') = \$^{r +qk'}$. By definition of $q$, this yields $\alpha(\mu(w)) = \alpha(\mu(w')) = \alpha(\$^r)$. As $\alpha$ recognizes~$H$, we get $\mu(w) \in H \Leftrightarrow \mu(w') \in H$, as desired.
\end{proof}

Theorem~\ref{thm:mdcov} provides log-space reductions from \md-covering to \abg-covering and from \md-separation to \grp-separation. Hence by Section~\ref{sec:abe}, \md-covering is in \conptime and by Section~\ref{sec:grp}, \md-separation is in~\ptime. In the next two subsections, we show that the \conptime upper bound for covering is tight, while \md-separation is in fact \nl-complete.

\subsection{Complexity of \md-covering}

As we explained above, \md-covering is in \conptime: Theorem~\ref{thm:mdcov} provides a logarithmic space reduction to \abg-covering which is itself in \conptime. Here, we prove that this upper bound is tight. We actually show that \emph{non} \md-coverability is \nptime-hard. More precisely, we present a logarithmic space reduction from \textup{3}-satisfiability (\tsat) to this problem. Given a \tsat formula $\varphi$, we explain how to construct a finite set of regular languages and show that it is not \md-coverable if and only if $\varphi$ is satisfiable. We only describe the construction: that \nfas for the regular languages in the set can be computed from $\varphi$ in logarithmic space is straightforward and left to the reader.

\begin{remark}
  Note that  Theorem~\ref{thm:mdcov} also provides a logarithmic space reduction from \md-covering to \grp-covering. Hence, the lower bound for \md-covering that we prove in this section transfers to \grp-covering, which is therefore \conptime-hard (recall that the upper bound for this problem is \pspace, since it amounts to checking nonemptiness of an intersection of automata).
\end{remark}

Let $C_1,\dots,C_k$ be the $3$-clauses such that $\varphi = \bigwedge_{i \leq k} C_i$ and let $x_1,\ldots,x_n$ be the propositional variables in $\varphi$. Consider the unary alphabet $U = \{\$\}$. We construct a finite set of regular languages over $U$. We encode the assignment of truth values for $x_1,\ldots,x_n$ by single words in $U^*$. Let $p_1,\ldots,p_n \in \nat$ be the first $n$ prime numbers. For each $w \in U^*$, we associate an assignment $val(w) \in \{0,1\}^n$ that encodes a mapping $x_i\mapsto b_i$ giving truth values for the variables $x_1,\ldots,x_n$. We define $val(w)=(b_1,\dots,b_n)$ such that for each $i \leq n$, we let $b_i = 1$ if $|w|$ is a multiple of $p_i$ and $b_i = 0$ otherwise. Note that since $p_1,\dots,p_n$ are primes, each assignment of truth values for $x_1,\ldots,x_n$ is encoded by some word in $U^*$. We now specify the regular languages associated to $\varphi$. For every $i \leq n$, we let,
\[
  \begin{array}{lllll}
    P_i & = & \{w \in U^* \mid |w| \equiv p_i \bmod 0\}, &&\\
    N_i & =& \{w \in U^* \mid |w| \not\equiv p_i \bmod 0\} &= &U^* \setminus P_i.
  \end{array}
\]
Finally, with every $j \leq k$, we associate a language $L_j$ to the $3$-clause $C_j$. By definition, we have $i_1,i_2,i_3 \leq n$ such that $C_j = \ell_{i_1} \wedge \ell_{i_2} \wedge \ell_{i_3}$ where $\ell_{i_k} \in \{x_{i_k}, \neg x_{i_k}\}$. For $k \in \{1,2,3\}$, we let $H_k = P_k$ if $\ell_{i_k} = x_{i_k}$ and  $H_k = N_k$ if $\ell_{i_k} = \neg x_{i_k}$. We then define $L_j = H_1 \cup H_2 \cup H_3$. One may verify from the definition that $L_j \in \md$ and that an \nfa recognizing $L_j$ can be computed from $\varphi$ in logarithmic space. Moreover, the following lemma may also be verified.

\begin{lemma} \label{lem:mdred}
  The language $\bigcap_{j \leq k} L_j$ consists of every word $w \in U^*$ such that the assignment $val(w)$ satisfies $\varphi$.
\end{lemma}

Since each assignment of truth values is encoded by some word in $U^*$, Lemma~\ref{lem:mdred} implies that $\varphi$ is satisfiable if and only if $\bigcap_{j \leq k} L_j\neq \emptyset$. Finally, since $L_1,\dots,L_k \in \md$, one may verify that $\bigcap_{j \leq k} L_j \neq \emptyset$ if and only if  $\{L_1,\dots,L_k\}$ is \emph{not} \md-coverable. Altogether, we obtain that $\varphi$ is satisfiable if and only if $\{L_1,\dots,L_k\}$ is \emph{not} \md-coverable: this is indeed a logarithmic space reduction from \tsat to non-coverability for \md.

\subsection{Complexity of \md-separation}

We now prove that \md-separation is in \nl, by an analysis the \grp-separation procedure for \emph{unary alphabets}. This implies that \md-separation is \nl-complete, as \nl is a generic lower bound for separation. Indeed, there exists a straightforward reduction from \nfa emptiness (which is \nl-complete) to \Cs-separation for an arbitrary Boolean algebra \Cs: given an \nfa~\As, $L(\As) = \emptyset$ if and only if $L(\As)$ is \Cs-separable from $A^*$.

Theorem~\ref{thm:mdcov} presents a log-space reduction from \md-separation to \grp-separation for languages over \emph{unary alphabets}. Hence, it suffices to prove that the latter problem is in \nl. Fix a single letter alphabet $A = \{a\}$. We prove that given as input two \nfas $\As_1$ and $\As_2$ over $A$, one may decide in \nl whether $L(\As_1)$ is \emph{not} \grp-separable from $L(\As_2)$. Since $\nl = \conl$ by the Immerman-Szelepcsényi theorem, this implies as desired that \grp-separation is in \nl for languages over unary alphabets. By Theorem~\ref{thm:grpsep}, the two following conditions are equivalent:
\begin{enumerate}
  \item $L(\As_1)$ is not \grp-separable from $L(\As_2)$.
  \item $L(\econs{\As_1}) \cap L(\econs{\As_2}) \neq \emptyset$.
\end{enumerate}
Therefore, we have to prove that the second condition can be decided in \nl. For $j \in \{1,2\}$, we write $\As_j = (Q_j,I_j,F_j,\delta_j)$. By definition, \econs{\As_j} is built from $\As_j$ by adding new transitions labeled by $a\inv$ (this is the construction of $\As_j \mapsto \acons{\As_j}$) and \veps-transitions (this is the construction of $ \acons{\As_j} \mapsto \econs{\As_j}$). It is standard that if we have \econs{\As_1} and \econs{\As_2} in hand, deciding whether $L(\econs{\As_1}) \cap L(\econs{\As_2}) \neq \emptyset$ can be achieved in \nl since this boils down to graph reachability (in the product of $\As_1$ and $\As_2$ whose set of states is $Q_1 \times Q_2$). Therefore, we have to prove that one may decide in \nl whether a given transition belongs to \econs{\delta_1} or \econs{\delta_2}.

This is immediate for the transitions labeled by $a \in A$ as they already belong to $\delta_1$ and~$\delta_2$. Let us now consider the transitions labeled by $a\inv \in A\inv$ which belong to \acons{\delta_1} and \acons{\delta_2}. By definition, for $j = 1,2$, and $q,r \in Q_j$, we have $(r,a\inv,q) \in \acons{\delta_j}$ if and only if $(q,a,r) \in \delta_j$ and $q,r$ are strongly connected. Clearly, this can be checked in \nl since testing whether $q,r$ are strongly connected boils down to graph reachability (which is in \nl). It remains to consider the $\veps$-transitions in \econs{\delta_1} and \econs{\delta_2}. We do so in the following lemma (this is where we use the hypothesis that the alphabet is unary).

\begin{lemma} \label{lem:unarybound}
  Let $j \in \{1,2\}$ and $q,r \in Q_j$, one may decide in \nl whether $(q,\veps,r) \in \econs{\delta_j}$.
\end{lemma}

\begin{proof}
  By definition, we have $(q,\veps,r) \in \econs{\delta_j}$ if and only if there exists $w  \in L_\veps \subseteq \tilas$ such that $w \in \lauto{\acons{\As_j}}{q}{r}$. Observe that since we have $A = \{a\}$, it follows that \mbox{$L_{\veps} = \{w \in \tilas \mid |w|_a\ = |w|_{a\inv}\}$}. We use this property to prove that deciding whether $(q,\veps,r) \in \econs{\delta_j}$ boils down to graph reachability, which again can be decided in \nl.

  We let $U = Q_j \times \integ$ be a set of vertices and consider the following set of edges:
  \[
    \begin{array}{llll}
      E &= & & \{((q,k),(q',k+1)) \mid (q,a,q') \in \acons{\delta_j}\} \\
        & & \cup &\{((q,k),(q',k-1)) \mid (q,a\inv,q') \in \acons{\delta_j}\}.
    \end{array}
  \]
  Consider the graph $G = (U,E)$. One may verify that \mbox{$(q,\veps,r) \in \econs{\delta_j}$} if and only if there exists a path from $(q,0)$ to $(r,0)$ in $G$. We prove that the latter condition can be checked in \nl. Let $V = \{(q,k) \in U \mid |k| \leq |Q_j|^2\}$.  We show that there exists a path $(q,0)$ to $(r,0)$ in $G$ if and only if there exists a path from $(q,0)$ to $(r,0)$ in $G$ using only states in $V$. It is then straightforward that this last property can be tested in \nl, since this is again a graph reachability problem over a graph with $|V| = |Q_j| \times (2|Q_j|+1)$ vertices, whose edges can be computed from $\As_j$ in \nl.

  The right to left implication is immediate. For the converse one, we consider a path from $(q,0)$ to $(r,0)$ in $G$. We prove that if this path contains a vertex in $U \setminus V$, then there exists a strictly shorter path from $(q,0)$ to $(r,0)$. One may then iterate the result to build a path that only contains states in $V$, completing the proof. Let $(s_0,k_0), \dots, (s_n,k_n) \in U$ be the vertices along our path: $(s_0,k_0) = (q,0)$, $(s_n,k_n) = (r,0)$, and for every $i \leq n$, we have  $((s_i,k_i),(s_{i+1},k_{i+1})) \in E$. Moreover, we know that there exits some index $h \leq n$  such that $(s_h,k_h) \not\in V$, \emph{i.e.}, such that $|k_h| > |Q_j|^2$. By symmetry, we assume that $k_h > |Q_j|^2$ and leave the case $k_h < -|Q_j|^2$ to the reader. We write $m = k_h$ for the proof. By definition of $E$ and since $k_0 = k_n = 0$, there exist,
  \[
    0 < i_1 < \cdots < i_{m-1} < h < i'_{m-1} < \cdots < i'_1 < n,
  \]
  such that $k_{i_1} = k_{i'_1} = 1,\dots, k_{i_{m-1}} = k_{i'_{m-1}} = m-1$. We also write $i_0 = 0$ and $i'_0 = n$. By hypothesis, $k_{i_0} = k_{i'_0} = 0$. Since $m > |Q_j|^2$, it now follows from the pigeonhole principle that there exist $0 \leq \ell_1 < \ell_2 \leq m-1$ such that $s_{i_{\ell_1}} = s_{i_{\ell_2}}$ and $s_{i'_{\ell_1}} = s_{i'_{\ell_2}}$. Let $\ell = \ell_2 - \ell_1$. One may verify from the definition of $E$ that the following paths exist in $G$:
  \[
    \begin{array}{l}
      (s_0,k_0) \rightarrow \cdots \rightarrow (s_{i_{\ell_1}},k_{i_{\ell_1}}) \rightarrow (s_{i_{\ell_2+1}},k_{i_{\ell_2+1}} - \ell)\\
      (s_{i_{\ell_2+1}},k_{i_{\ell_2+1}} - \ell) \rightarrow \cdots \rightarrow (s_h,k_h - \ell)\\ (s_h,k_h - \ell) \rightarrow \cdots \rightarrow (s_{i'_{\ell_2-1}},k_{i'_{\ell_2-1}} - \ell) \rightarrow (s_{i'_{\ell_1}},k_{i'_{\ell_1}}) \\
      (s_{i'_{\ell_1}},k_{i'_{\ell_1}}) \rightarrow \cdots \rightarrow (s_n,k_n).
    \end{array}
  \]
  Altogether, we get a strictly shorter path from $(q,0)$ to $(r,0)$, which completes the proof.
\end{proof}

\section{Conclusion}
\label{sec:conc}
We proved simple separation and covering algorithms for the classes \grp, \abg and \md using only standard notions from automata theory. For \grp and \abg, the proofs are based on the automata-theoretic construction ``$\As\mapsto\acauto$''. Since the statements behind the two algorithms (\emph{i.e.}, Theorem~\ref{thm:grpsep} and Theorem~\ref{thm:abesep}) are  similar, a natural question is whether their proofs can be unified (as of now, they are independent).
We also obtained tight complexity bounds:  separation is \nl-complete for \md, co-\nptime-complete for \abg and \ptime-complete for \grp. Covering is co-\nptime-complete for both \md and \abg, and between co-\nptime and \pspace for \grp. This raises the question of the exact complexity of \grp-covering.

\bibliographystyle{abbrv}

\bibliography{main}

\end{document}